%% file: Data_Drive_FDI_Journal_full.tex
\documentclass[journal]{IEEEtran}
\usepackage{graphicx,epsfig,amsmath,amssymb,bm}
\usepackage{amssymb,balance}
\usepackage{amsmath}
\usepackage{amsthm} 
\usepackage{tikz}
\usepackage{amssymb}
\usepackage{algorithmicx}
\usepackage{algpseudocode}
\input{macros}

\usepackage{bbm}
\usepackage{subcaption}
\usepackage{tabularx}
\usepackage{multirow}
\usepackage{verbatim}

\newtheorem{lemma}{Lemma}
\newtheorem{theorem}{Theorem}

\usepackage[ruled,linesnumbered]{algorithm2e}

\newcommand{\beqa}{\begin{eqnarray}}
\newcommand{\eeqa}{\end{eqnarray}}
\newcommand{\dsp}{\displaystyle}

\ifCLASSINFOpdf

\else

\fi

\begin{document}

\title{Data-Driven False Data Injection Attacks Against Power Grids: A Random Matrix Approach }
\author{Subhash~Lakshminarayana~\IEEEmembership{Member, IEEE}, Abla Kammoun~\IEEEmembership{Member, IEEE}, M\'erouane Debbah~\IEEEmembership{Fellow, IEEE} and H. Vincent Poor~\IEEEmembership{Fellow, IEEE} 
\thanks{S. Lakshminarayana is with the University of Warwick, Coventry, UK (email: subhash.lakshminarayana@warwick.ac.uk). A. Kammoun is
with the Electrical Engineering Department, King Abdullah University of
Science and Technology, Thuwal, Saudi Arabia (abla.kammoun@gmail.com).
M. Debbah is with the Mathematical and Algorithmic Sciences Lab, Huawei
Technologies Co. Ltd., France (merouane.debbah@huawei.com). H. Vincent Poor is with the Department of Electrical Engineering, Princeton University, USA (email: poor@princeton.edu). The work was partially presented at ICASSP-2018 \cite{LakshICASSP2018}. This research was supported in part by a Startup grant at the University of Warwick and in part by the U.S. National Science Foundation under Grants  DMS-1736417 and ECCS-1824710.
}}

\maketitle

\begin{abstract}
We address the problem of constructing false data injection (FDI) attacks
that can bypass the bad data detector (BDD) of a power grid. The attacker is assumed to have access to only power grid measurement data traces collected over a limited period of time and no other prior knowledge about the grid. Existing related algorithms are formulated under the assumption that the attacker has access to measurements collected over a long (asymptotically infinite) time period, which may not be realistic. We show that these approaches do not perform well when the attacker has access to measurements from a limited time window only. We design an enhanced algorithm to construct FDI attack vectors in the face of limited measurements that can nevertheless bypass the BDD with high probability. The algorithm design is guided by results from random matrix theory. Furthermore, we characterize an important trade-off between the attack's BDD-bypass probability and its sparsity, which affects the spatial extent of the attack that must be achieved. Extensive simulations using data traces collected from the MATPOWER simulator and benchmark IEEE bus systems validate our findings.
\end{abstract}

\begin{IEEEkeywords}
Data-driven FDI attack, random matrix theory, spiked model, sparse false data injection attack.
\end{IEEEkeywords}

\IEEEpeerreviewmaketitle

\section{Introduction}

The growing integration of information and communication technologies (ICTs) in power grids has made them vulnerable to cyber attacks \cite{Ukraine2016}. In this work, we study the problem of constructing false data injection (FDI) attacks against state estimation in a power grid from an attacker's perspective. It has been shown \cite{Liu2011} that if the attacker obtains detailed knowledge of the power grid topology and transmission line reactance values -- i.e., the system's {\em measurement matrix} -- then they can construct FDI attacks that bypass the grid's bad data detector (BDD). Subsequent research \cite{LiPoorScaglione2013,Kim2015,YuBlind2015,ChinBlindAC2018} has shown that an attacker can learn the power grid's measurement matrix \cite{LiPoorScaglione2013}, or learn the structure of its column space by estimating the basis vectors \cite{Kim2015,YuBlind2015,ChinBlindAC2018} from accessed measurement data (i.e., nodal power injections and line power flows) only. The focus of our work is on constructing these \emph{data-driven} FDI attacks. Undetected FDI attacks can severely affect the power grid operation, such as unsafe voltage/frequency excursions \cite{LaksheEnergy2017}.

Prior work on designing data-driven BDD-bypass attacks \cite{Kim2015,YuBlind2015,ChinBlindAC2018} are based on the technique of \emph{principal component analysis} (PCA), whose basic idea is to use the \emph{sample covariance matrix} to identify the eigenmodes along which the data exhibits the greatest variance. It performs efficiently when the measurement period is significantly large compared to the signal dimension \cite{anderson1963}.  However, data-driven attacks based on the PCA-based method fail to remain stealthy if the measurements accessed by the attacker have missing values (due to communication loss or device malfunctions). Subsequent work \cite{Anwar2016} applied a robust-PCA based approach to deal with this issue. Similar robust learning techniques were also applied in \cite{TianDataDriven2018, Tian2019} to deal with the joint problem of missing measurements and the construction of sparse FDI attacks. A different stream of work \cite{RahmanIncomplete2012, LiuIncomplete2015} has considered the problem of crafting FDI attacks when the attacker has incomplete/partial knowledge of the power grid topology and transmission line reactances. However, these works do not consider the attacker's learning of the grid parameters from the measurement data. Hence, they do not present a complete picture of the attacker's capabilities in this context. 


A major drawback of existing work on data-driven FDI attacks \cite{Kim2015,YuBlind2015,ChinBlindAC2018, Anwar2016, TianDataDriven2018, Tian2019} is that they perform well only when the attacker has access to measurements from a large time window (asymptotically infinite). For practical purposes, it is important to understand these attacks under a limited measurement time window. The reasons include (i) active topology control \cite{DFACTS2007} or renewable energy integration that leads to an inherently dynamic operating environment, thereby rendering measurements outdated and irrelevant after some time; and (ii) an attacker's desire or need (e.g., due to limited resources or limited exploitation time windows) to launch the attack quickly. Thus, in a practical scenario, the measurement time period may not be asymptotically large compared to the signal dimension, especially for large power grids (refer to the example presented in Section \ref{subsec:drawback}). It has been demonstrated that under such a regime, the principal component estimated by PCA is inconsistent \cite{Johnstone2009}. Indeed, our experiments show that FDI attacks constructed by the existing PCA-based algorithms \cite{Kim2015,YuBlind2015,ChinBlindAC2018} do not perform well (in terms of the BDD-bypass probability) when applied in a limited measurement period setting.

To address these shortcomings, in this paper, we analyze the problem of finding a BDD-bypassing attack using measurement data collected from a limited time window (comparable to the measurement signal dimension) and identify guiding principles for the solution in this context. The analysis provides an important understanding of the attacker's capabilities in designing FDI attacks by accessing the system measurements. The understanding has practical relevance in the design of defense strategies, such as determining the frequency of reactance perturbations in the context of moving target defense \cite{LakshDSN2018}, which in turn depends on the attacker's capability of learning the system parameters.

Under the limited measurement period setting, a key issue is that only a few eigenmodes can be reliably estimated from the sample covariance matrix. This number, in turn, depends on the length of the measurement period relative to the signal dimension. To bypass the BDD with a high probability, it is important for the attacker to identify these critical eigenmodes. Direct application of the PCA method as in \cite{Kim2015,YuBlind2015,ChinBlindAC2018} does not use this knowledge, and hence, performs poorly. In this work, we propose an enhanced algorithm to construct FDI attacks in the face of limited measurement period that can nevertheless bypass the BDD with high probability.
Our algorithm design is based on results from random matrix theory (RMT). The application is based on an important observation that the power grid's state estimation utilizes several redundant sensor measurements to filter the effect of measurement noise and obtain an optimal estimate on the system state \cite{AburExposito2004, wood1996power}. In other words, the dimension of the measurement vector is much greater than the size of the system state. 

Under this setting, the covariance matrix of the sensor measurements has a structure similar to the so-called ``spiked models" in RMT \cite{BaikSilver06, Paul07}, which comprises of a low-rank perturbation of a scaled identity matrix. Here, the leading few eigenmodes correspond to the subspace spanned by the signal (i.e., system state), whereas the bulk of the eigenmodes (corresponding to the redundant measurements) correspond to the noise subspace. For data obtained from the spiked model, RMT results can be used to characterize the number of eigenmodes that can be estimated accurately as a function of the measurement time window, as well as the corresponding estimation accuracy \cite{BaikSilver06, Paul07}. Using these results, the attacker can carefully design the attack vector by restricting it to a lower-dimensional subspace that is spanned by the accurately estimated eigenmodes only and bypass the BDD with a high probability. Otherwise, the inaccurately estimated basis vectors may mislead the attack vector to a subspace that is different from the intended one, thereby risking detection by the BDD.

However, restricting the attack vector to a lower-dimensional subspace of the estimated column space limits the attacker's freedom in crafting the FDI attack. In particular, a resource-constrained attacker may wish to minimize the number of the meters that must be compromised, or equivalently find the sparsest attack vector in the execution \cite{OzaySparse2013,KimPoor2011}. Clearly, maximizing the sparsity of the attack vector is best achieved if we have an unconstrained choice of this vector over the full estimated column space of the measurement matrix. Hence, the attacker faces a fundamental tradeoff. On the one hand, as we observed, restricting the attack vector to a lower-dimensional subspace (spanned by the accurately estimated basis vectors) will enhance the BDD-bypass probability under the limited measurement period setting; i.e., the restriction makes the attack efficient temporally. On the other hand, this restriction may reduce the sparsity of the optimized attack vector, thus making it less efficient spatially. 
To understand the tradeoffs between the conflicting objectives, we compute the sparsest attack vector while constraining it to subspaces of varying lower dimensions of the full estimated column space.

To summarize, the contributions of this work are as follows.

\begin{itemize}

	\item We propose an enhanced algorithm to construct data-driven FDI attacks in the limited measurement period setting that can bypass the BDD with high probability using results from RMT. 

	\item We characterize an important trade-off between the FDI attack's BDD-bypass probability and the number of power meters in the grid that the attacker has to compromise in achieving the attack. 

	\item We illustrate the fundamental trade-off by performing extensive simulations using benchmark IEEE bus systems. The results show that the attacker can significantly enhance the BDD-bypass probability using our proposed approach.

\end{itemize}

To the best of our knowledge, this work is the first to apply RMT results in the context of smart grid security. While RMT results have found wide application in other domains such as wireless communications, finance, physics etc. (we refer the reader to reference \cite{Couillet2011}, Chapter 1 for a comprehensive review of RMT applications), its application to smart grids has been scarce. In particular, the application of the RMT spiked model results to FDI attack construction is novel and has not been considered previously, and this is one of the important contributions of our work.

The rest of the paper is organized as follows. We describe the system model in Section~\ref{sec:II}. We review existing subspace method based algorithms to construct data-driven FDI attacks and point our their drawbacks in Section~\ref{III}. In Section~\ref{sec:Limited_Mes}, we present data-driven FDI attacks under the limited measurements period setting using RMT results and analyze its performance. The trade-offs in data-driven FDI attacks are discussed in Section~\ref{sec:Trade_off}. The simulation results are presented in Section~\ref{sec:Sim_Res}. Finally, conclusions are drawn in Section \ref{V}. The technical proofs are presented in Appendices A, B and C.

\emph{Notations}: Throughout this work, we use boldface lowercase
and uppercase letters to designate column vectors and
matrices, respectively. For a matrix $\Am,$ we let $\av_i$ denote its $i^{\text{th}}$ column. The notation $\Am_s$ denotes a matrix consisting of the first $s$ columns of the matrix $\Am,$ i.e., $\Am_s =  [\av_1,\dots,\av_s],$ for any integer value $s.$

\section{System Model}
\label{sec:II}
We consider a power grid that is characterized by a set of buses
$\mathcal{N} = \{0,1,2,\dots,N\}$ and transmission lines 
$\mathcal{L} = \{1,2,\dots,L\}$. The node with index $0$ is used to represent the reference node.
The grid is assumed to operate in a time slotted manner indexed by $t = 1,2,\dots,T.$ 
To model power flows within the grid, we adopt the direct current (DC) power flow model \cite{wood1996power}. Under this model, the system state corresponds to the nodal voltage phase angles, 
which we denote by $\thetav[t] = [\theta_1[t],\dots,\theta_N[t]]^T;$ i.e., $\theta_i[t], \ i \in \mathcal{N}$ is the voltage phase angle at bus~$i$ during the time slot $t.$  For the reference bus, $\theta_0[t] = 0, \forall t.$
We assume that the system state fluctuates around a mean value, e.g., due to the temporal variations of the load. Thus, $\thetav[t]  = \bar{\thetav} + \epsilonv[t],$ where $\epsilonv[t]$ is assumed to be an independent and identically distributed (i.i.d.) random vector (across time) whose covariance is given by $\sigma^2_{\theta} \Id,$ (where $\Id$ denotes an identity matrix). Here in, $\bar{\thetav}$ represents the bus voltage phase angles due at a base load (e.g., obtained by solving the optimal power flow considering a base load). The temporal independence assumption of the system state fluctuations can be met by taking measurements with sufficient load/angle variations over time.

\subsubsection*{State Estimation and Bad Data Detection}
The system state $\thetav[t]$ is monitored using sensors deployed at the buses and transmission lines. These sensors measure respectively the nodal power injections and the forward/reverse line power flows. Under the linear power flow model, these measurements, which we denote by $\zv[t] \in \mathbb{R}^{M}$ (where $M$ denotes the number of measurements), are related to the system state $\thetav[t] \in \mathbb{R}^{N}$ as
\begin{equation}
\zv[t] = \mathbf{H} \thetav[t] + \nv[t], \quad t = 1, 2, \cdots, T,  \label{eqn:eqzt}
\end{equation}
where $\mathbf{H} \in \mathbb{R}^{M \times N}$ is the measurement matrix and $\nv[t]$ is the sensor measurement noise. The noise is assumed to be zero-mean Gaussian\footnote{We note that the analysis in this paper is more generally applicable to any distribution of the noise as long as the distribution has a bounded fourth moment.} with covariance matrix $\sigma^2_n \Id$, and independent of the system state $\thetav[t]$. It is also assumed to be i.i.d. across the time slots. Without the loss of generality we set $\sigma^2_n = 1$ throughout the paper, and scale the covariance of the $\thetav[t]$ appropriately (i.e., we set $(\sigma^{\prime}_{\theta})^2 = (\sigma^2_{\theta}/\sigma^2_n)$ in our analysis).  
The measurement matrix $\Hm$ depends on the system topology (i.e., the bus connectivity) and the branch reactances \cite{wood1996power}. We assume that within the considered time interval $T$, $\Hm$
does not change.   
The estimate of the system state, denoted by $\widehat{\thetav}[t],$ is recovered from the measurement vector $\zv[t]$ using a maximum-likelihood (ML) technique \cite{AburExposito2004}: $\widehat{\thetav}[t] = \left( \mathbf{H}^{T}  \mathbf{H}\right)^{-1} \mathbf{H}^{T} \zv[t].$

After state estimation, the residual vector is given by $\rv[t] =  \zv [t] - \mathbf{H} \widehat{\thetav}[t].$
The BDD checks for possible measurement inconsistencies in $\zv[t]$ works by comparing the norm of the residual vector $r = ||\rv[t]||^2$ against a pre-defined threshold $\zeta.$ It raises an alarm if $r \geq \zeta.$ The threshold $\zeta$ is selected to ensure a certain false-positive (FP) rate.

\subsubsection*{Attacker Model} 
We consider an attacker who can eavesdrop on the measurement data communicated between the field devices and the control center by exploiting vulnerabilities in the communication system. We consider different \emph{read} and \emph{write} capabilities for the attacker, since read access only requires passive sniffing of the network data whereas write access requires modifying the network packets (which is typically harder to perform in practice). Accordingly, we assume that the attacker can \emph{read} all the measurements in the system (i.e., no missing measurements). However, the attacker may have \emph{write} access to only a partial subset of measurements (see Section V). Furthermore, the attacker has no other information about the grid (e.g., its topology or transmission line reactances).

The attacker's objective is to craft FDI attacks against the state estimation. Denote the attack vector by $\av[t] \in \RR^M,$ the sensor measurements under attack by $\zv_a[t],$ where $\zv_a[t] = \zv[t]+\av[t]$, and the BDD residual under attack 
by $r_a[t] = ||\zv_a[t] - \mathbf{H} \widehat{\thetav}_a[t] ||^2.$ It has been shown \cite{Liu2011} that for an attack of the form $\av[t] = \Hm \cv[t],$ the residual value remains unchanged under the attack, i.e., $r_a[t] = r[t].$ Hence, the BDD's detection probability for such attacks is no greater than the FP rate. We will henceforth refer to these attacks as \emph{undetectable} attacks. Note that constructing an undetectable FDI attack requires the knowledge of $\Hm.$ In data-driven FDI attack, the attacker strives to construct an undetectable FDI attack by learning the system parameters using the accessed measurement data.

\section{Subspace Method Based Algorithm and the Drawbacks}
\label{III}
In this section, we review existing subspace-based approaches for constructing undetectable data-driven FDI attacks~\cite{Kim2015,YuBlind2015,ChinBlindAC2018}, and point out its drawbacks under a practical regime of a limited observation time window.

\subsection{Algorithm Description} 

Note that designing an undetectable attack is equivalent to finding a non-zero vector in $Col(\mathbf{H})$, or equivalently, a linear combination of the basis vectors that span $Col(\mathbf{H}).$ The attacker must estimate the basis vectors using the noisy measurement data $\zv[t], \ t = 1,\dots,T.$ This problem is well studied in the signal processing literature \cite{Krim1996}, and has been used to guide the construction of data-driven FDI attacks. 

The key idea is to use the covariance matrix of the measurements $\Sigmam_{\zv} = \mathbb{E} [(\zv[t]-\mathbb{E} [\zv[t]]) (\zv[t] - \mathbb{E} [\zv[t]])^T].$ From \eqref{eqn:eqzt}, it follows that
\begin{align}
\bm{\Sigma}_{\zv} = \sigma^2_{\theta} \mathbf{H}  \mathbf{H}^T+  \mathbf{I}.
\end{align}
Let $\Um \Lambdam \Um^T $ be the SVD of $\bm{\Sigma}_{\zv},$ 
where $\Um = [\uv_1,\dots,\uv_M],$
is a matrix consisting of  the eigenvectors of $\Sigmam_{\zv},$ and 
$\Lambdam = \text{diag} (\lambda_1,\dots,\lambda_M)$ is a matrix consisting of the eigenvalues of $\Sigmam_{\zv}.$
Note that the rank of the matrix $\sigma^2_\theta {\Hm} {\Hm}^T$ is $N.$ Thus, the first $N$ columns of $\Um$ corresponding to the $N$ largest singular values must form the basis
vectors of $Col(\sigma^2_\theta {\Hm} {\Hm}^T).$ Since, 
$Col(\sigma^2_\theta {\Hm} {\Hm}^T)$ is equivalent to $Col(\Hm)$, they also form the basis vectors of $Col(\Hm)$ \cite{Krim1996}. Thus, the attacker must estimate the eigenvectors of $\Sigmam_{\zv}$ in order to construct an undetectable FDI attack vector.

We note that the attacker cannot directly execute the procedure stated above since the $\bm{\Sigmam}_{\zv}$ is unknown.
However, it can be estimated using the measurement data $\{ \zv [t] \}^T_{t = 1}$. Based on this observation, the procedure to construct data-driven FDI attacks is summarized in Alg.\ref{alg1}. (We use the superscript $\widehat{}$ to denote estimates of the corresponding quantities. The notation $\Am_s$ denotes a matrix consisting of the first $s$ columns of the matrix $\Am,$ i.e., $\Am_s =  [\av_1,\dots,\av_s],$ for any integer value $s.$ ).
\begin{algorithm}
\caption{Data-driven FDI attack}
\label{alg1}
 \begin{itemize}
\item[{1.}] Using measurements $\{ \zv[1],\dots,\zv[T] \},$ compute the sample
covariance matrix $\widehat{\bm{\Sigma}}_{\zv}$ as
\begin{align*} 
\widehat{\bm{\Sigma}}_{\zv} = \frac{1}{T-1} \sum_{t = 1}^{T} \left( \zv[t] - \bar{\zv} \right)\left( \zv[t] - \bar{\zv} \right)^T, 
 \end{align*}
where $\bar{\zv}$ denotes the sample mean given by $\bar{\zv} = \frac{1}{T-1} \sum_{t = 1}^{T} \zv[t] .$
%
\item[{2.}] Perform singular value decomposition (SVD) of $\widehat{{\Sigmam}}_{\zv}$ as
$\widehat{\bm{\Sigma}}_{\zv} = \widehat{\mathbf{U}} \widehat{\bm{\Lambda}} \widehat{\mathbf{U}}^T,$ where $\widehat{\Um} = [\widehat{\uv}_1,\dots,\widehat{\uv}_M]$ and $\widehat{\Lambdam} = \text{diag} (\widehat{\lambda}_1,\dots,\widehat{\lambda}_M).$

\item[{3.}] Construct an undetectable FDI attack vector as
$\mathbf{a}[t] = \widehat{\mathbf{U}}_N \mathbf{c}[t],$
where $\cv[t] \in \RR^N.$
\end{itemize}
\end{algorithm}

\subsection{Drawbacks of Existing Techniques}
\label{subsec:drawback}
The aforementioned subspace estimation algorithm performs well in a classical setting when the number of temporal measurements are far greater than the system dimension (i.e., $T \gg M, M/T \to 0$).
However, under a practical setting, it is unreasonable to expect the availability of an ``infinite time window", especially for large bus systems. For instance, consider the IEEE-118 bus system which has $M = 490$ measurements per time slot (assuming a fully measured system). Under an optimistic assumption of obtaining a temporally decorrelated measurement every minute, the attacker would require a measurement time window of $4900$ minutes, or approximately, $80$ hours, for the ratio of $M/T = 0.1.$ However, the system topology may have changed well before this duration. 

Thus, we focus on a practically relevant scenario where for a given bus system, the size of the measurement vector $M$ and the measurement time window $T$ are of the same orders of magnitude, i.e., $M/T = p$. This scenario is especially relevant for large power grids. Under this regime, the principal components estimated by Algorithm~1 are known to be inconsistent \cite{Johnstone2009}. Thus, in the rest of the paper, we present an enhanced algorithm for strengthening the attack's BDD-bypass probability when the attacker has access to measurements from a limited time window. Furthermore, we
characterize an important trade-off between the attack's BDD-bypass probability and the number of compromised measurements in executing the attack.

\section{Data-driven FDI Attacks with Limited Number of Measurements}
\label{sec:Limited_Mes}
In this section, we present an enhanced algorithm for designing data-driven FDI attacks. From the discussion in Section~III, note that the problem at hand is equivalent to estimating the principal eigenvalues/vectors of $\Sigmam_{\zv}$ from the corresponding sample covariance matrix $\widehat{\Sigmam}_{\zv}.$ Under a limited measurement period setting, a key issue is that only a few eigenmodes can be reliably estimated from the sample covariance matrix. RMT results can help us identify those key eigenmodes as well as characterize their estimation accuracy. We first present a brief overview of RMT and then show its application to data-driven FDI attacks.

\subsection{Brief Introduction to Random Matrix Theory and the Spiked Model}
RMT studies the properties of matrices whose entries are random. Of particular interest are the matrix's spectral properties when its dimensions grow large.

\subsubsection*{Marcenko-Pastur Law}
Consider a matrix $\Xm = [\xv[1],\xv[2],\dots,\xv[T]]\in \RR^{M \times T}$ whose columns $\xv[1],\xv[2],\dots$ are drawn from a multivariate Gaussian distribution with zero mean and identity covriance matrix, i.e., $\xv[1] \sim \mathcal{N} ({\bf 0}, \Id), i = 1,\dots,T.$ When the number of snapshots $T$ is very large and the vector size $M$ is fixed, i.e., $T \to \infty$ and $M/T \to 0,$ then the sample covariance matrix converges to the true covariance matrix ($\Id$ in this case) asymptotically \cite{Bill86}:
\begin{align}
\widehat{\Sigmam}_{\xv} = \frac{1}{T} \sum^T_{t = 1} \xv[t] \xv[t]^H \asto \Id \defines \mathbb{E} [\xv[t] \xv[t]^H].
\end{align}
The convergence result above holds for any matrix norm, i.e., $||\Id - \widehat{\Sigmam}_{\xv}|| \to 0$ on a set of probability one. Further, the eigenvalues of $\widehat{\Sigmam}_{\xv}$ will converge to a single mass at $1.$

However, when the number of snapshots $T$ is large, but not extremely large compared to the vector size $M,$  i.e., $T \to \infty$ and $M/T = p > 0$, the above result no longer holds. Specifically, $||\Id - \widehat{\Sigmam}_{\xv}||$ does not go to zero despite element-wise convergence of $\widehat{\Sigmam}_{\xv}$ to the identity matrix. 
This can be observed in Fig.~\ref{fig:Spiked_model} (top figure), where the histogram of eigenvalues of $\widehat{\Sigmam}_{\xv}$ is plotted for $M = 500, T = 2000.$ Note that the eigenvalues of $\widehat{\Sigmam}_{\xv}$ do not converge to a single mass at $1,$ but are spread around $1.$ This paradoxical behavior occurs since despite being large, $T$ is never very large compared to $M.$ The distribution of the eigenvalues of $\widehat{\Sigmam}_{\xv}$ in this case converges to a non-random distribution known as the Marcenko-Pastur (MP) law \cite{Marcenko_1967}, which has a probability density function given by
\begin{align}
f(x) = (1-p^{-1}) \delta (x) + \frac{1}{2 \pi c x} \sqrt{(x-a)^+ (b-x)^+},
\end{align}
where $a = (1 - \sqrt{p})^2,$ $b = (1 + \sqrt{p})^2$ and $\delta(x)$ is the Dirac-delta function. We note that $a$ and $b$ mark the extremities of the spread of the eigenvalues around $1.$ Moreover, for $T \to \infty, M/T = p >0,$ it is guaranteed that no eigenvalue of $\widehat{\Sigmam}_{\xv}$ is found outside the set $[a,b],$ almost surely.

\subsubsection*{Spiked Model}

\begin{figure}[!t]
	\centering
	\begin{subfigure}{0.33\textwidth}
		\includegraphics[width=1\textwidth]{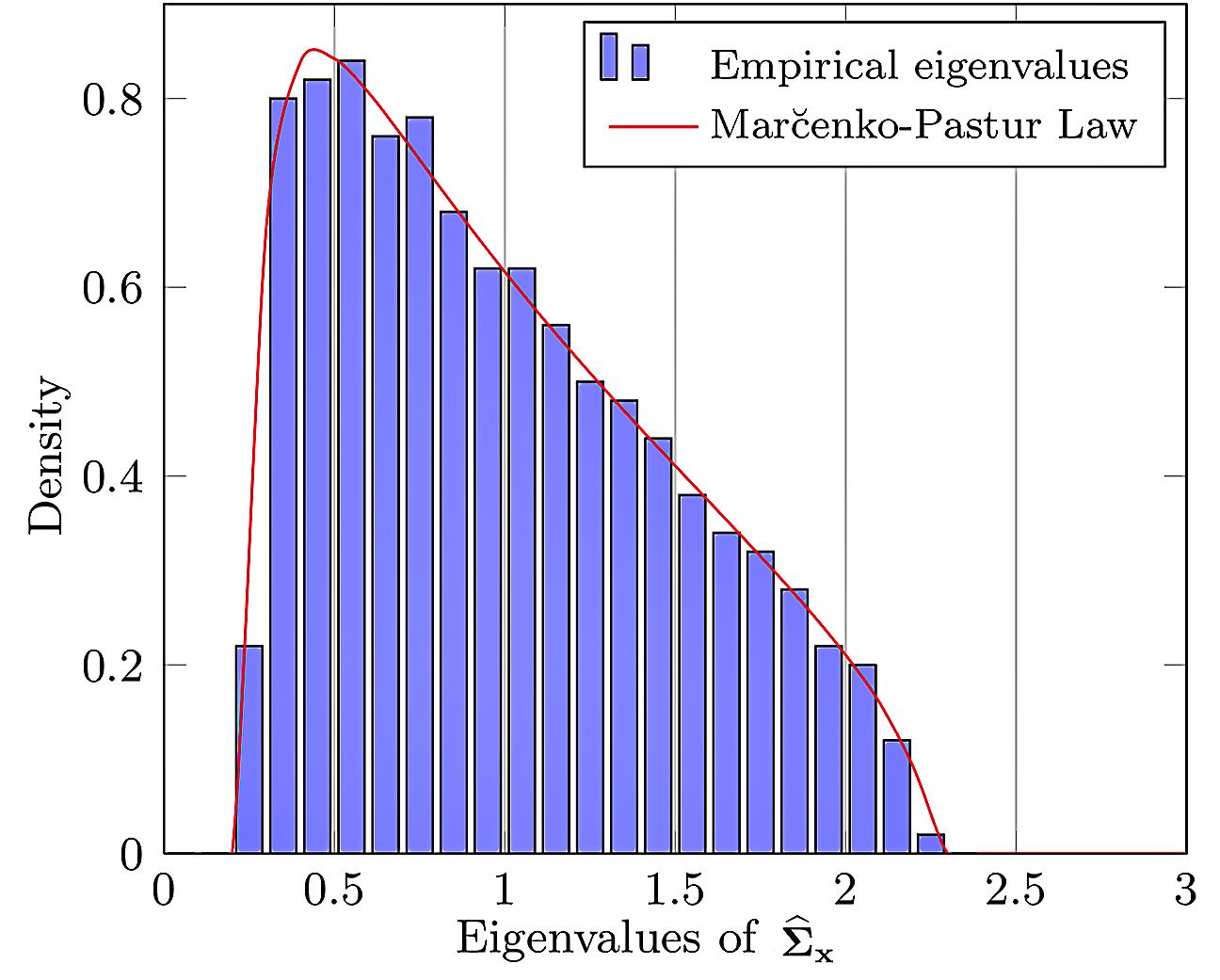}
	\end{subfigure}
	~
	\begin{subfigure}{0.33\textwidth}
		\includegraphics[width=1\textwidth]{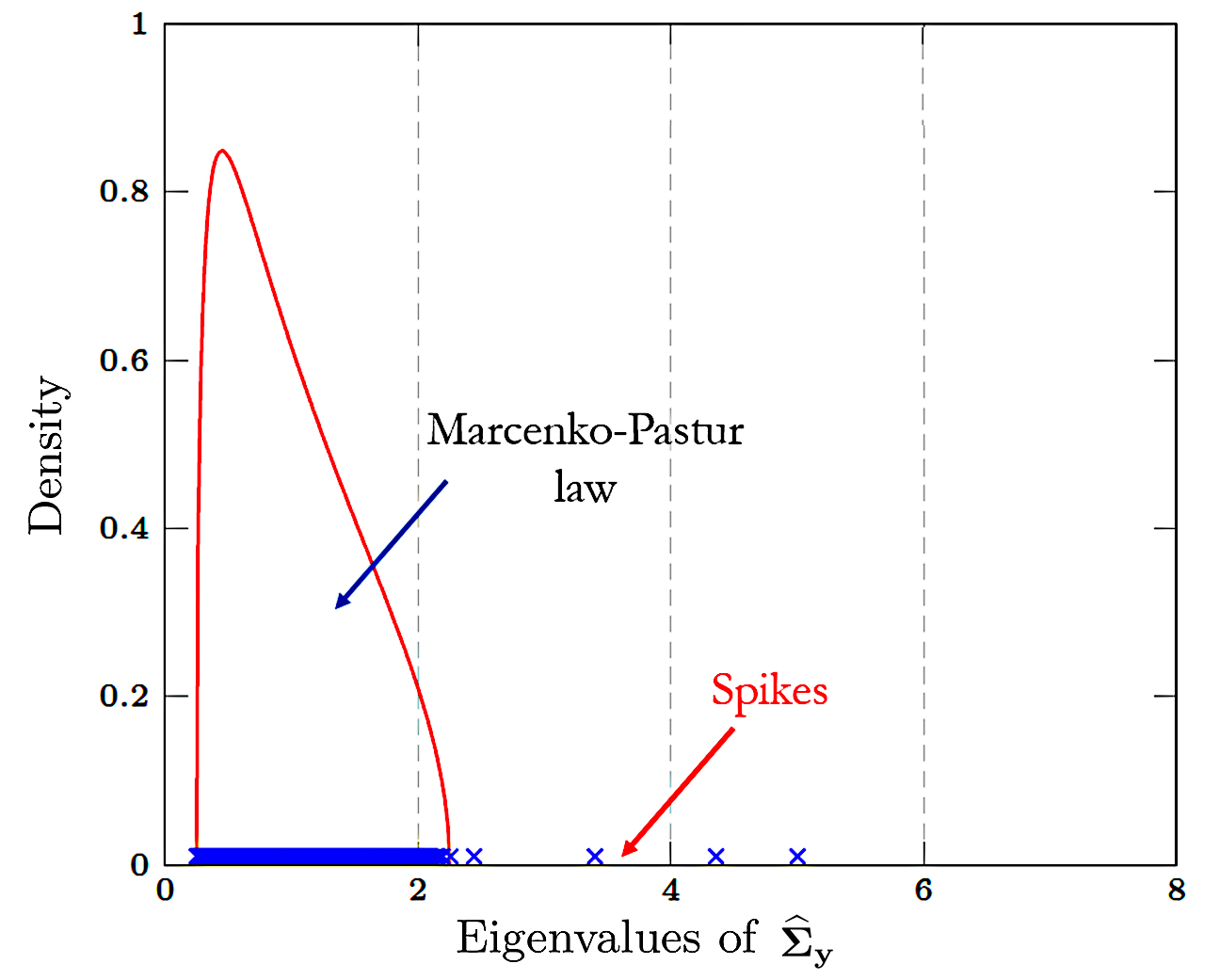}
	\end{subfigure}
	\caption{Top: Marcenko-Pastur law, Bottom: RMT spiked model. Figure due to \cite{Couillet2011}.}
	\label{fig:Spiked_model}
	\vspace*{-0.4 cm}
\end{figure}

Next, we describe the spiked model in RMT. 
Consider a matrix $\Ym = [\yv[1],\yv[2],\dots,\yv[T]]\in \RR^{M \times T}$ whose columns $\yv[1],\yv[2],\dots$ are drawn from the multivariate Gaussian distribution with zero mean and covariance matrix given by 
\begin{align}
\Sigmam_{\yv} =  \Id + \sum^N_{i = 1} \mu_i \uv_i \uv_i^H. \label{eqn:spike}
\end{align}
where $N$ is a fixed quantity. Compared to $\xv_i,$ the covariance matrix of $\yv_i$ is perturbed by $\sum^N_{i = 1} \mu_i \uv_i \uv_i^H.$ Under the spiked model, $\mu_1,\dots,\mu_N$ are referred to as  ``spike" eigenvalues. 
Of particular interest is a scenario when $M, T \to \infty, M/T = p$ and the number of spikes $N$ is small in comparison to $M$ and $T.$ The eigenvalues of $\Sigmam_{\yv}$ are given by  
\begin{align}
[\underbrace{\mu_1+1 ,\dots,\mu_N+1}_{N \ \text{terms}}, \underbrace{1, \dots, 1}_{M-N \ \text{terms}}]. \label{eqn:Evals}
\end{align}
Since $N$ is small compared to $M,$ we note that bulk of the eigenvalues of $\Sigmam_{\yv}$ are $1$ and a few eigenvalues exceed $1.$
Now, consider the eigenvalues of the sample covariance matrix
$\widehat{\Sigmam}_{\yv} = \frac{1}{T} \sum^T_{t = 1} \yv[t] \yv[t]^H.$
Since the bulk of the eigenvalues of $\Sigmam_{\yv}$ are $1,$  the majority of eigenvalues of $\widehat{\Sigmam}_{\yv}$ will lie within the MP distribution (i.e., between the extremities specified by $a$ and $b$). However, one would expect the ``leading $N$ eigenvalues" of $\widehat{\Sigmam}_{\yv}$ (corresponding to the eigenvalues $\mu_i+1$  of ${\Sigmam}_{\yv}$) to be found outside the distribution of the MP law (see Fig.~\ref{fig:Spiked_model}). Surprisingly, the number of eigenvalues that can be found outside $f$ depends critically on the ratio $p = M/T.$ This result was formalized in \cite{BaikSilver06} and stated here:
\begin{theorem}
	\label{thm:spike_result}
Consider $\yv[t] = \mathcal{N}(0,\Sigmam_{\yv}) \in \RR^{M \times 1},$ where $\Sigmam_\yv$ is defined in \eqref{eqn:spike}.
Let 
$\widehat{\Um} \widehat{\Lambdam} \widehat{\Um}^T$ denote the eigenvalue decomposition of $\widehat{\Sigmam}_{\yv},$ where $\widehat{\Um} = [\widehat{\uv}_1,\dots,\widehat{\uv}_M],$ and $\widehat{\Lambdam} = \text{diag} (\widehat{\lambda}_1,\dots,\widehat{\lambda}_M).$ Assume $N$ is fixed, and independent of $M$ and $T.$
Then, when  $M,T \to \infty, M/T = p,$  for all $\mu_i > \sqrt{p},$ with probability one, $\widehat{\lambda}_i \geq  (1+\sqrt{p_N})^2$ 
and 
\begin{align}
\Big{|} \widehat{\lambda}_i - 1 - \mu_i - \frac{p (1+\mu_i)}{\mu_i} \Big{|} \asto 0.
\end{align}
Moreover,  $| \mu_i - \widehat{\mu}_i| \asto 0,$ where $\widehat{\mu}_i$ can be obtained from $\widehat{\lambda}_i$ as
\begin{align}
\widehat{\mu}_i = \frac{\widehat{\lambda}_i + 1 - p + \sqrt{(\widehat{\lambda}_i +1 - p)^2 - 4 \widehat{\lambda}_i}}{2} - 1. \label{eqn:mu_hat_defn}
\end{align}
Further, for all $\mu_i > \sqrt{p},$ we also have 
\begin{align}
\Big{|} \widehat{\bf u}_i^T{\bf u}_j{\bf u}_j^{T}\widehat{\bf u}_i-  \frac{1 -p/\mu_i^2}{1+ p/\mu_i} \delta_{i=j} \Big{|} \asto 0, i,j = 1,\dots,s. \label{eqn:proj_rmt}
\end{align}
\end{theorem}
The main idea of Theorem~\ref{thm:spike_result} is the following. Consider 
\begin{align}
\mu_1 > \mu_2 > \dots > \mu_s > \sqrt{p},  \qquad 1 \leq s \leq N \label{eqn:muvec},
\end{align}
where $s \leq N$ is the number of spike eigenvalues that are greater than $\sqrt{p}.$ Then, the result \cite{BaikSilver06} states that for all $\mu_i > \sqrt{p},$ when $M,T \to \infty, M/T = p > 0$, there exists a deterministic and one-to-one mapping between eigenvalue of the sample covariance matrix ($\widehat{\Sigmam}_{\yv}$), i.e.,  between $\widehat{\lambda}_i$ and $\mu_i.$ In other words, all $\mu_i$ which satisfy $\mu_i > \sqrt{p}$ can be recovered from the eigenvalues of the sample covariance matrix. A similar result also holds for estimating the corresponding eigenvectors \cite{Paul07}, i.e., the corresponding eigenvectors (for which $\mu_i > \sqrt{p}$) can be reliably recovered from the eigenvectors of the sample covariance matrix (see Theorem \ref{thm:spike_result}). It is important to note that for eigenmodes corresponding to $\mu_i < \sqrt{p},$ these relationships do not hold, and the corresponding eigenvalue/vectors cannot be recovered. Thus, the quantity $\sqrt{p}$ represents a fundamental ``phase transition'' point in estimating the spike eigenvalues/vectors from the sample covariance matrix.

\subsection{Application of Spiked Model To Data-Driven FDI Attack} 
We now discuss the application of RMT spiked model results to our problem. From \eqref{eqn:eqzt}, the covariance matrix of the measurements $\Sigmam_{\zv}$ can be expressed through eigen decomposition as
\begin{align}
\Sigmam_{\zv} =  \Id + \sum^N_{i = 1} \mu_i \uv_i \uv_i^H.
\end{align}
where $ \{ \mu_i \}^N_{i = 1} $ denote the eigenvalues of $\sigma^2_{\theta} \Hm  \Hm^H$ and $ \{ \uv_i \}^N_{i = 1} $ the corresponding eigenvectors. The result of RMT spiked model is important in the context of data-driven FDI attack, since it precisely characterizes the information about $Col(\Hm)$ that the attacker can recover from the measurements as a function of the observation time window $T$ (specifically, the ratio $p = M/T$). To construct a data-driven FDI that can bypass the BDD with a high probability, the attacker must first estimate the number of eigenvalues/vectors, $s,$ that can be reliably recovered from the measurements $ \{ \zv[t] \}^T_{t = 1}$. Note that the attacker cannot directly use \eqref{eqn:muvec} to determine $s,$ since the value of $\mu_i$ is not known. Using the result of Theorem \ref{thm:spike_result}, it follows that for $\mu_i > \sqrt{p},$ with probability $1,$ we have $\widehat{\lambda}_i >  (1+\sqrt{p})^2$. 
Thus the attacker can determine $s$ by counting the number of eigenvalues of the sample covariance matrix that exceed $ (1+\sqrt{p})^2$ 
, i.e., 
\begin{align}
s = \{ \#i,  \widehat{\lambda}_i  >  (1+\sqrt{p})^2 \}. \label{eqn:num_eig}
\end{align}
Note that the direct application of the subspace estimation algorithm as proposed in \cite{Kim2015,YuBlind2015,ChinBlindAC2018} (Algorithm 1) uses all $N$ estimated eigenmodes for the construction of the FDI attack. However, following the application of RMT spiked model results, it is clear that eigenmodes for which $\mu_i < \sqrt{p}$ cannot be recovered from the sample covariance matrix, and hence, must not be used in the construction of FDI attack. 


After determining $s,$ the attacker can construct a data-driven FDI attack as $\av = \widehat{\Um}_s \cv_s,$ where  $\cv_s \in \RR^s$ denotes an $s-$dimensional vector. In particular, the vector $\cv_s$ can be tuned by the attacker to achieve his objectives, such as minimizing the attack's detection probability or causing the desired attack impact. In the rest of this section, we describe how the attacker can achieve these objectives using results from RMT. 

We first focus on attack detection probability. To this end, we characterize the BDD residual with a data-driven FDI attack. 
\begin{lemma}
	\label{lem:proj_diag}
	For a data-driven FDI attack $\av = \widehat{\Um}_{s} \cv_s,$ the BDD residual with attack, $r_a,$ follows a   non-central $\chi^2$ distribution with $M-N$ degrees of freedom and a non-centrality parameter $\nu$ given by 
	\begin{align}
	\nu = {\bf c}_s^{T}{\bf c}_
s-{\bf c}_s^{T} \widehat{\bf U}_s^{T}{\bf U}_N{\bf U}_N^{T}\widehat{\bf U}_s{\bf c}_s. \label{eqn:nu}
	\end{align}
	 For $M,T \to \infty, M/T = p$ the second term of the right hand side of $\eqref{eqn:nu}$ converges to 
	\begin{align}
	{\bf c}_s^{T} \widehat{\bf U}_s^{T}{\bf U}_N{\bf U}_N^{T}\widehat{\bf U}_s{\bf c}_s  - {\bf c}_s^{T}  \Omegam_s  {\bf c}_s \asto 0, \label{eqn:nu_hat}
	\end{align}
	where $\Omegam_s = \text{diag} (\omega_1,\dots,\omega_s)$ and
	\begin{align}
	\omega_i = \frac{1 -p/\mu_i^2}{1+ p/\mu_i}, i = 1,\dots,s. \label{eqn:omega_def}
	\end{align}
	Further, the attacker can obtain a consistent estimator $\widehat{\omega}_i$ of  ${\omega}_i$ as 
	$| \omega_i - \widehat{\omega}_i| \asto 0, i = 1,\dots,s$ where,
	\begin{align}
	\widehat{\omega}_i &= \frac{1 -p/\widehat{\mu}_i^2}{1+ p/ \widehat{\mu}_i}, i = 1,\dots,s, \label{eqn:omega_hat} 
	\end{align}
	and $\widehat{\mu}_i$ as in \eqref{eqn:mu_hat_defn}.
Thus, it follows that
\begin{align}
\nu  - {\bf c}_s^{T}  (\Id - \widehat{\Omegam}_s)  {\bf c}_s \asto 0, \label{eqn:nu_hat_final}
\end{align}
where $\widehat{\Omegam}_s = \text{diag} (\widehat{\omega}_1,\dots,\widehat{\omega}_s).$
\end{lemma}
\begin{proof}
The proof is omitted due to the lack of space and can be found in Appendix~A, Part I.
\end{proof}
The result of \eqref{eqn:nu} along with the asymptotic approximations \eqref{eqn:nu_hat}-\eqref{eqn:omega_hat} provides a tractable expression for the attacker to compute the detection probability for a given data-driven FDI attack $\av = \widehat{\Um}_{s} \cv_s$. Observe that all the quantities required to compute the asymptotic approximation of $\nu$ depend on the estimated parameters only (i.e., $\widehat{\lambda}_i$). Using these expressions, the attacker can tune $\cv_s$ to minimize the attack's detection probability. We analyze the results further. 

Note that the entries of the matrix $\widehat{\bf U}_s^{T}{\bf U}_N{\bf U}_N^{T}\widehat{\bf U}_s$ represent the projection of the eigenvectors of the sample covariance matrix $\widehat{\Sigmam}_{\zv}$ onto the eigenvectors of the population covariance matrix ${\Sigmam}_{\zv}$. In particular, the result of Lemma~\ref{lem:proj_diag} states that asymptotically, the estimated eigenvectors 
$\widehat{\uv}_i$ are orthogonal to $\uv_j, j \neq i$, since $\Omegam_s$ is diagonal. Specifically, $|\widehat{\uv}^T_i \uv_i|^2 \asto \omega_i$ and $|\widehat{\uv}^T_i \uv_j|^2 \asto 0, i \neq j.$ The following lemma illustrates the relationship between the projections.
\begin{lemma}
\label{lem:dec_order}
The diagonal elements $\{\omega_i \}_{i = 1}^s$ and  $\{\widehat{\omega}_i \}_{i = 1}^s$ follow
$1 > \omega_1 \geq \omega_2 \geq,\dots,\geq \omega_s > 0$ and $1 > \widehat{\omega}_1 \geq \widehat{\omega}_2 \geq,\dots,\geq \widehat{\omega}_s > 0$ respectively.
\end{lemma}
\begin{proof}
The proof can be found in Appendix~B.
\end{proof}
From Lemma~\ref{lem:proj_diag} and Lemma~\ref{lem:dec_order}, it follows that the projection of $\widehat{\uv}_i$ onto $\uv_i$ is in the decreasing order of the eigenmode index. Note that minimizing the detection probability is equivalent to minimizing the non-centrality parameter $\nu$ of the $\chi^2$ distribution. From \eqref{eqn:nu_hat_final}, it follows the attacker can compute ${\bf c}_s$ which minimizes ${\bf c}_s^{T} (\Id-\widehat{\Omegam}_s){\bf c}_s.$ 
However, directly minimizing this expression would result
in a trivial solution $\cv_s = {\bf 0},$  (i.e., a zero attack). Thus, we must constrain the attack impact in order to obtain a meaningful attack. 

We quantify the attack impact in terms of the second norm of the error in state estimate (for the system operator) due to the FDI attack. Specifically, we let $\widehat{\thetav}^a$ denote the estimate of the system state from measurements with FDI attack, $\zv^a.$ Then, 
$\Delta \thetav \defines \widehat{\thetav}-\widehat{\thetav}^a$ is the error in the state estimate due to the FDI attack. Using this, the data-driven FDI attack can be formulated as the following optimization problem:
\beqa
&\dsp \min_{\cv_s} &   \cv_s^T (\Id- \widehat{\Omegam}_s) \cv_s \label{eqn:opt_init} \\ 
& \text{s.t.} &  ||\Delta \thetav||^2_2  \geq \tau \nonumber 
\eeqa 
In the optimization problem \eqref{eqn:opt_init}, the attacker designs $\cv_s$ to minimize the probability of detection among all attacks that satisfy $||\Delta \thetav||^2_2  \geq \tau.$ However, \eqref{eqn:opt_init} cannot be solved by the attacker directly, as $||\Delta \thetav||^2_2$ depends on the measurement matrix $\Hm,$ that is unknown to the attacker. To address this issue, we present a consistent estimate of $||\Delta \thetav||^2_2$ in the large system regime that depends only on the attacker's estimated parameters  in the following lemma:
\begin{lemma}
	\label{lem:state_error}
For $M,T \to \infty,$ the quantity  $||\Delta \thetav||^2_2$ converges to
\begin{align}
||\Delta \thetav||^2_2 - \sigma^2_{\theta}\cv_s^T  \Mm^{-1} \Omegam \cv_s \asto 0.
\end{align}
Further, we have 
\begin{align}
\sigma^2_{\theta} \cv_s^T  \Mm^{-1} \Omegam \cv_s  - \widehat{\sigma}^2_{\theta} \cv_s^T  \widehat{\Mm}^{-1} \widehat{\Omegam} \cv_s \asto 0, \label{eqn:state_conv_det}
\end{align}
where $\Mm = \text{diag} (\mu_1,\dots,\mu_N) $ and $\widehat{\Mm} = \text{diag} (\widehat{\mu}_1,\dots,\widehat{\mu}_N).$
\end{lemma}
\begin{proof}
The proof is presented in Appendix~A, Part II. 
\end{proof}
Note from Lemma~\ref{lem:state_error} that the asymptotic approximation of $||\Delta \thetav||^2_2$ depends on the estimate of the variance of the system state $\widehat{\sigma}^2_{\theta}.$ The attacker can estimate this by monitoring historical fluctuations of the system load (note that this is a second-order statistic and hence, need not be estimated in real-time).

Based on the result of Lemma~\ref{lem:state_error}, optimization problem \eqref{eqn:opt_init} can be reformulated as follows:
\beqa
&\dsp \min_{\cv_s} &   \cv_s^T (\Id- \widehat{\Omegam}_s) \cv_s \label{eqn:opt_final} \\ 
& \text{s.t.} & \widehat{\sigma}^2_{\theta} \cv_s^T    \widehat{\Mm}^{-1} \widehat{\Omegam}_s \cv_s  \geq \tau \nonumber 
 \eeqa 
 The solution to \eqref{eqn:opt_final} can be characterized in closed form and its result leads to the following theorem:
\begin{theorem}
\label{thm:opt_attack}
For $M,T \to \infty, M/T = p,$ the optimal data-driven FDI attack that solves \eqref{eqn:opt_final}  is given by
\begin{align}
\av = \sqrt{\frac{\tau}{ \widehat{\sigma}^2_{\theta} (\widehat{\omega}_1/\widehat{\mu}_1)}} \widehat{\uv}_1.
\end{align}
\end{theorem}
\begin{proof}
The proof of Theorem~1 follows by noting that the solution to optimization problem \eqref{eqn:opt_final} is given by $c_1 = \sqrt{\frac{\tau}{ \widehat{\sigma}^2_{\theta} (\widehat{\omega}_1/\widehat{\mu}_1)}}$ and $c_2 = c_3= \dots = c_s = 0.$
The details are presented in Appendix~C.
\end{proof}
Theorem~\ref{thm:opt_attack} implies that the attacker can minimize the detection probability by aligning the attack vector along
 $\widehat{\uv}_1$ while achieving the desired attack impact. Coincidentally, from Lemma~\ref{lem:dec_order}, $\widehat{\uv}_1$ is also the most accurately estimated eigenmode. 

Theorem~\ref{thm:opt_attack} also implies that the optimal attack must be restricted to a $1-$dimensional subspace of the estimated space $\widehat{\Um}_N$. A natural question is whether there a cost to pay for this restriction? We will address this question in the next section where we consider the attack's sparsity in addition to the factors considered in this section.

\subsection{Discussion}
A few comments are in order regarding our results. First, we note from Theorem \ref{thm:spike_result} that the application of spiked model requires $N$ to be fixed and independent of $M$ and $T.$ Strictly speaking, the power grid model does not satisfy this condition, since $N$ (dimension of the state vector) also grows for large grids. However, despite this limitation, we will show by simulations in Section~\ref{sec:Sim_Res} that RMT spiked model results are accurate for various power grid bus configurations as long as the number of sensor measurements $M$ is large compared to $N.$ In other words, our results match closely when there are a significant number of redundant measurements, which is reasonable for the state estimation problem \cite{AburExposito2004, wood1996power}. Thus, the RMT spiked model can be used for analysing the data-driven FDI attacks. 

 Second, our results assume that the sensor measurement noises have identical variances, i.e., $\Sigmam_n = \mathbb{E} [\nv[t] \nv[t]^T] = \sigma^2 \Id.$ When they are not identical, i.e., when $\Sigmam_n =\diag [\sigma^2_1,\sigma^2_2,\dots,\sigma^2_M],$ then the attacker must modify the results (of Lemma~\ref{lem:proj_diag} and Lemma~\ref{lem:state_error}) using those from a generalized spiked model as in \cite{Honghetero2018}.

\section{Trade-offs in Data-Driven FDI Attacks}
\label{sec:Trade_off}
The analysis considered thus far in this paper only focusses on attacker's learning of $Col(\Hm)$. The learning phase only requires the attacker to obtain \emph{read} access to the sensor measurements. However, executing the FDI attack requires the attacker to modify the sensor measurements values, which in turn requires \emph{write} access. A graphical illustration is presented in Fig.~\ref{fig:readwrite}. Note that from an attacker's point of view, \emph{read} access to sensor measurements is easier to obtain compared to \emph{write} access, since it only involves passive sniffing of the network data, where as write access requires modification of the network packets. 
Thus, resource-constrained attackers may wish to minimize the number of sensors they must compromise to execute the FDI attack, or equivalently, maximize the attack's sparsity. 

\begin{figure}[!t] 
	\centering{\includegraphics[width=0.4\textwidth]{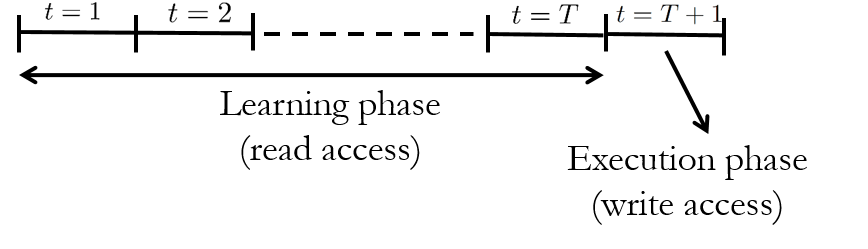}}
	\caption{Learning and execution phase of data-driven FDI attack.}
	\label{fig:readwrite}
	\vspace*{-0.6 cm}
\end{figure}
Recall that from our analysis in the previous section, the optimal data-driven FDI attack is one that is restricted to a $1-$ dimensional subspace of the estimated space. Restricting the attack to a lower-dimensional subspace makes it hard to enforce sparsity. Naturally, maximum sparsity of the attack vector can be achieved if we have an unconstrained
choice of the attack vector over the full estimated column space. On the other hand, using the inaccurately estimated basis vectors will increase the attack's detection probability (recall that the first basis is estimated most accurately, followed by the second, etc. refer to Lemma~\ref{lem:dec_order}). Thus, the attacker faces a fundamental trade-off between the attack's BDD-bypass probability and the attack's sparsity. In particular, the number of estimated eigenmodes for constructing the FDI attack must be chosen to balance between the two factors.

To formalize this trade-off, we cast the FDI attack construction as a sparse optimization problem for all $m = 1,\dots,s$ as follows:
\beqa
K^*_m = &\dsp \min_{\cv_m} &   \| \widehat{\Um}_{m} \cv_m \|_0, \label{eqn:opt_sparse} \\ 
& \text{s.t.} &  \cv_m^T    \widehat{\Mm}^{-1} \widehat{\Omegam} \cv_m \geq \tau\nonumber.
 \eeqa 
The objective function of \eqref{eqn:opt_sparse} gives the number of non-zero elements in the FDI attack vector while restricting
the attack vector to an $m-$dimensional subspace of the estimated column space, where $m \leq s$ (\eqref{eqn:num_eig}). 
The optimization problem \eqref{eqn:opt_sparse} can 
be solved using an $l_1-$relaxation based approach. We omit the details here and refer the reader to \cite{KimPoor2011}.
We illustrate the trade-off by simulations in Section~\ref{sec:Sim_Res}.

\section{Simulation Results}
\label{sec:Sim_Res}
In this section, we present the simulation results. 
All the simulations are conducted using the MATPOWER simulator \cite{Zimmerman2011}. Unless stated otherwise, the simulations are conducted using an IEEE-14 bus system considering the DC power flow model of \eqref{eqn:eqzt}. As in standard DC state estimation, we consider the forward and reverse power flows and nodal power injection measurements. For the IEEE-14  bus system, the number of measurements $ M = 2 \times L + N  = 54$ (recall $L$ is the number of links and $N$ is the number of nodes). We consider Gaussian noise with standard deviation $\sigma_n = 0.02$ pu (approximately $1-2 \% $ of the full-scale measurement).
The eigenmodes are estimated following the steps 1 and 2 of Algorithm~1. 
The measurement data is generated according to \eqref{eqn:eqzt}, where $\thetav[t] = \bar{\thetav} + \epsilonv[t].$ Here in,  $\bar{\thetav}$ is obtained by solving the optimal power flow formulation considering base load values provided in the MATPOWER case file. The fluctuations  $\epsilonv[t]$ are assumed to be i.i.d. Gaussian random vectors with standard deviation $\sigma_{\theta} = 0.002$ pu (i.e., $\sigma_{\theta}/\sigma_n = 0.1$).
The FDI attacks are constructed using the estimated eigenmodes and their detection probability is computed by averaging the BDD's detection results over $1000$ independent trials. The BDD threshold is adjusted such that the FP rate is set to $0.02$. The results are presented next.

\subsection{Eigenmode Estimation Accuracy}
First, we examine the estimation accuracy of different eigenmodes by evaluating the projection metric $|\widehat{\uv}^T_i \uv_i|^2$. We also verify the accuracy of the RMT approximation in  Theorem~\ref{thm:spike_result}. To this end, we compare $|\widehat{\uv}^T_i \uv_i|^2$ obtained from simulations with $\widehat{\omega}_i$ computed according to \eqref{eqn:omega_hat}. For the measurement time window $T,$ we consider two regimes, (i) a non-asymptotic regime with $T = 0.5 M$ and (ii) an asymptotic regime with $T = 100M.$ 
The results are plotted in Fig.~\ref{fig:proj_est} by averaging across $1000$ trials. The bars represent mean values over the trials and the vertical lines (on top of the bars) represent the fluctuation around this mean value. We make the following observations.

Firstly, in the non-asymptotic regime, the estimation accuracies of the different eigenmodes vary. In particular, they are arranged in the decreasing order of the eigenmode index. This is consistent with our observation in Lemma~\ref{lem:dec_order}. 
Secondly, it can be observed that the RMT approximations $\widehat{\omega}_i$ (blue bars) are reasonably accurate, though there is a non-zero but negligible gap between the simulations and RMT results. The gap exists due to the fact that the number of spikes in the power grid model are large and equal to the dimension of the state vector (see the discussion in Section IV-A). However, despite this limitation, the gap is small and the RMT results are a good approximation. 
Thirdly, recall that RMT approximations only exist for the eigenmodes $i \leq s$ where $s$ is computed according to \eqref{eqn:nu_hat}. The value of $s$ for each of the simulation cases is indicated in the figure description. It can be observed that the estimation accuracy for eigenmodes beyond this value of $s$ is poor. Hence, they must not be utilized for FDI attack construction as prescribed by our analysis based on the RMT spiked model. 
Finally, we observe that in the asymptotic regime however, i.e.  $T = 100M$, all the eigenmodes can be estimated with a high accuracy and Algorithm~1 can be used directly for the design of FDI attack.

\begin{figure}[!t]
	\centering
	\begin{subfigure}{0.45\textwidth}
		\includegraphics[width=1\textwidth]{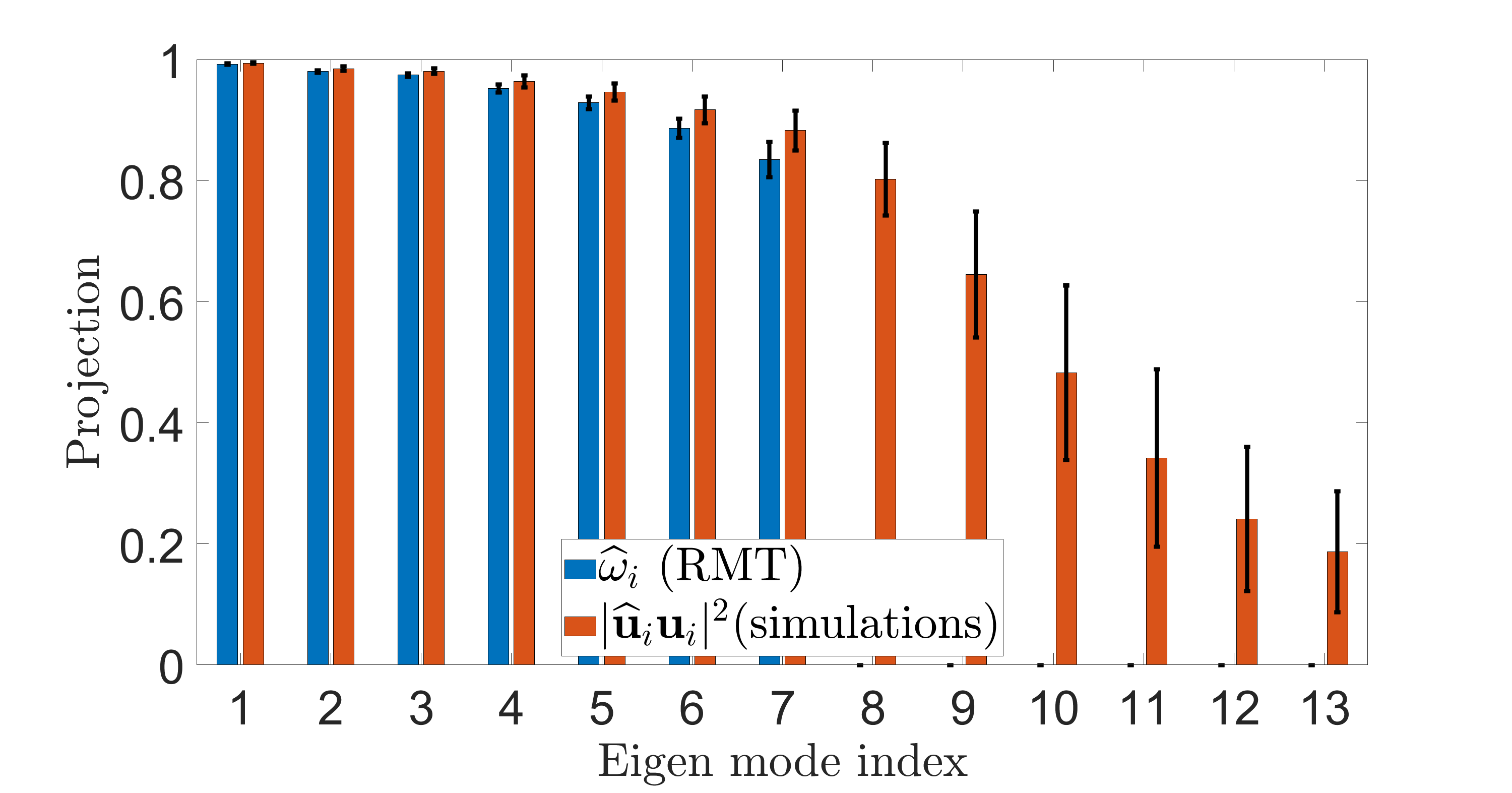}
	\end{subfigure}
	~
	\begin{subfigure}{0.45\textwidth}
		\includegraphics[width=1\textwidth]{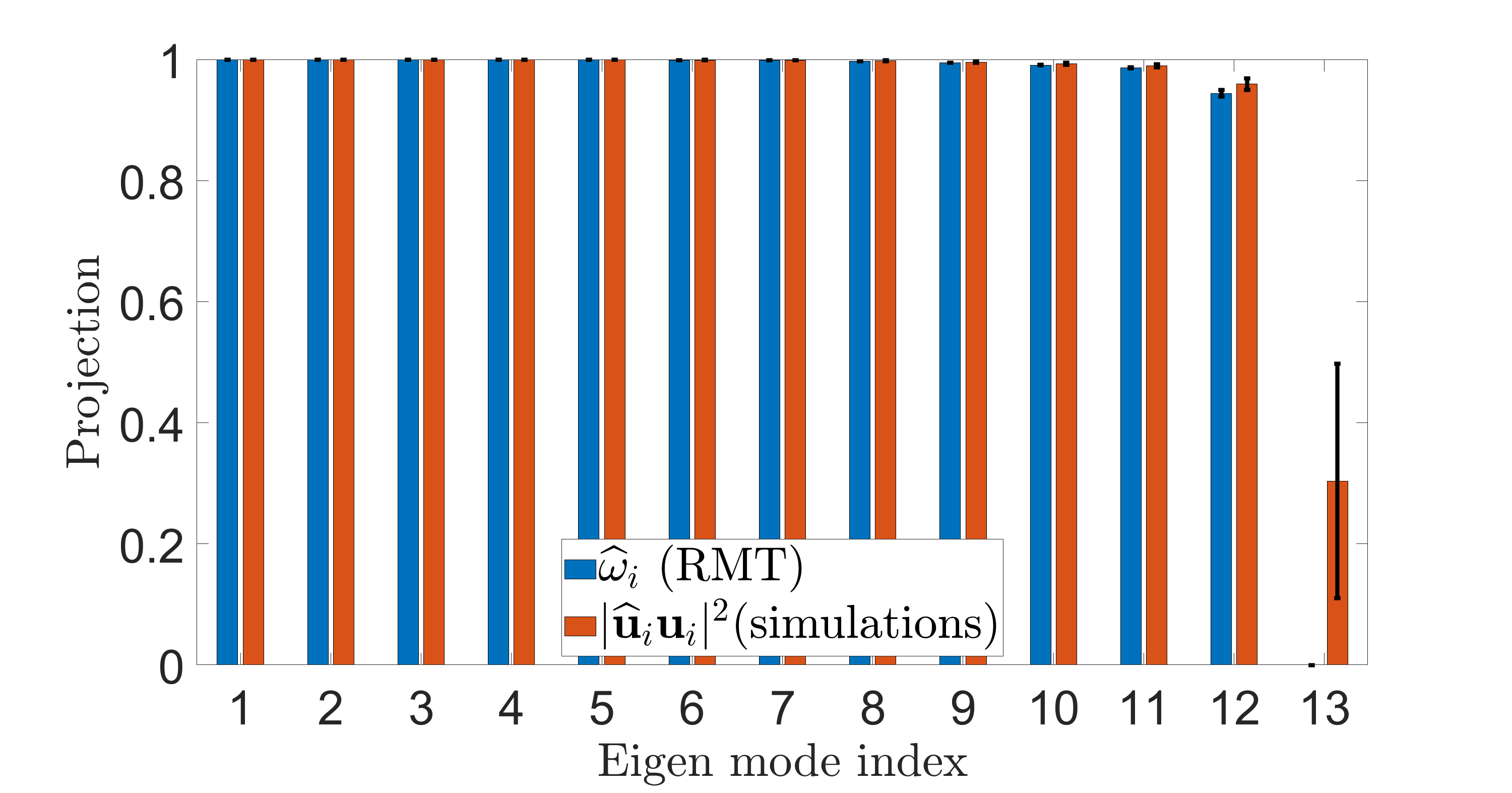}
	\end{subfigure}
	\caption{Eigenmode estimation accuracy using simulations and RMT approximation. Top: non-asymptotic regime, $p = 0.5$ ($s = 7$), Bottom: asymptotic regime, $p = 0.005$ ($s = 12$).  }
	\label{fig:proj_est}
	\vspace*{-0.4 cm}
\end{figure}

\subsection{Detection Probability of Data-Driven FDI Attacks} 
Next, we examine the detection probability of FDI attacks constructed using different estimated eigenmodes. Specifically, for each estimated eigenmode $i,$ the FDI attack is constructed as $\av =  c_i \widehat{\uv}_i,$ where $c_i$ is set to $c_i = \sqrt{\frac{\tau}{ \widehat{\omega}_i/\widehat{\mu}_i}},$ such that it satisfies the constraint of \eqref{eqn:opt_final}. Recall that asymptotically this ensures that $||\Delta \thetav||^2_2  \geq \tau$. The value of $\tau$ is set to $0.3.$ This causes an average normalized state estimation error (across trials), measured as $\eta = \frac{|| \widehat{\thetav}_a - \thetav ||_2}{|| \widehat{\thetav} - \thetav||_2}$, of $4.$ Note that $\eta$ represents the increase in state estimation due to the attack (herein, a $4$ times increase in the state estimation error).
We conduct $1000$ simulation trials and plot the results in Fig.~\ref{fig:Det_prob}. The RMT approximations of the detection probability are also plotted in Fig.~\ref{fig:Det_prob}. They are computed by evaluating the $\mathbb{P} (X \geq \tau),$ where $X$ is a $\chi^2$ distributed random variable with $M-N$ degrees of freedom and a non-centrality parameter $\nu = {\bf c}_s^{T}  \widehat{\Omegam}_s  {\bf c}_s$ (following the result of Lemma \ref{lem:proj_diag}). Once again, we make the following observations.
Firstly, the detection probability increases with the eigenmode index and the attack $\av =c_1  \widehat{\uv}_1 $ has the lowest detection probability, confirming the result of Theorem~\ref{thm:opt_attack}. Secondly, it can be observed for $i \geq s,$ the detection probability becomes very high, thus confirming the phase transition phenomenon of the RMT spiked model. 
Finally, the detection probability becomes lower as we increase the training time $T$. 

\begin{figure}[!t]
	\centering
		\includegraphics[width=0.45\textwidth]{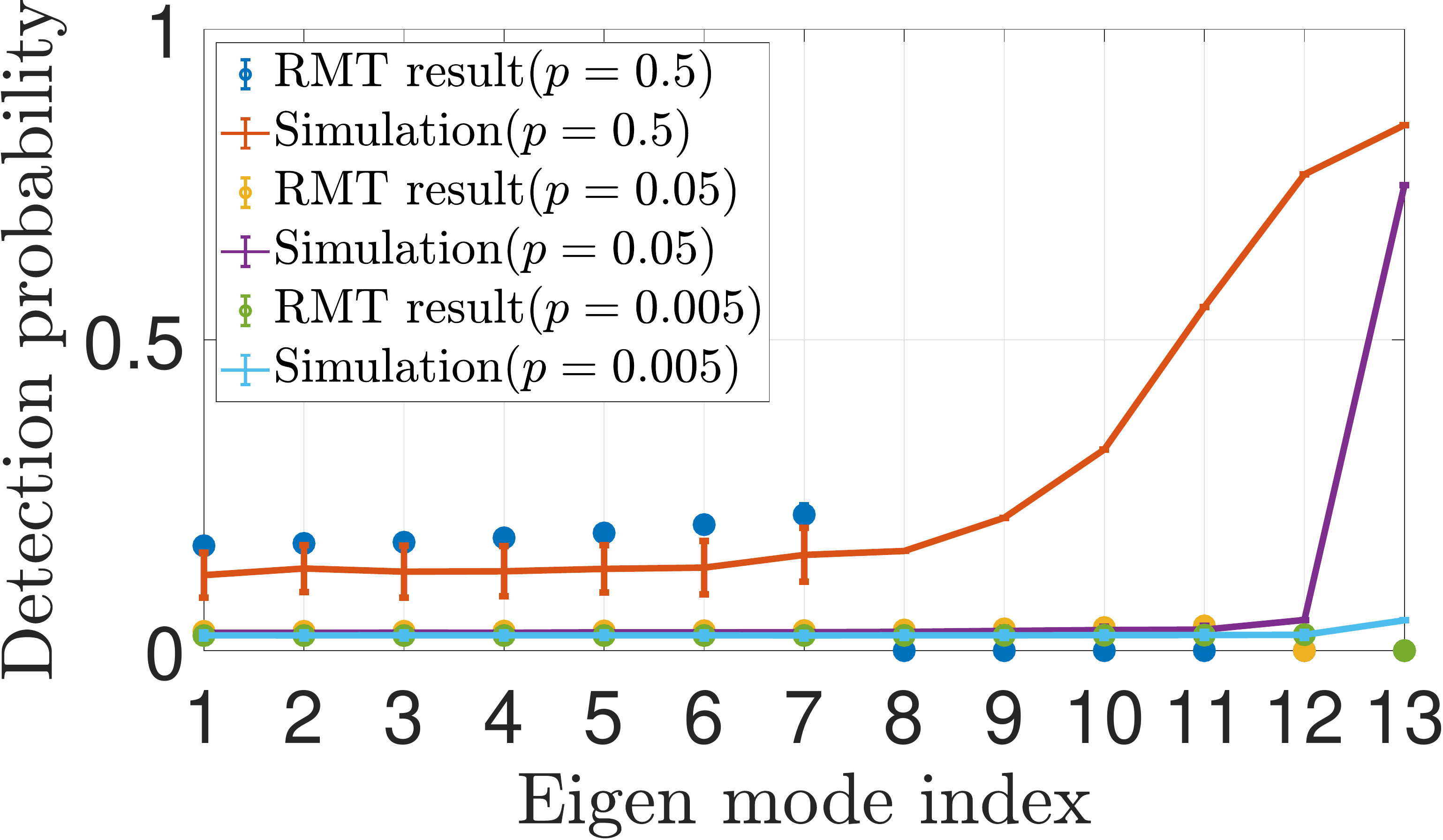}
	\caption{Detection probability of FDI attacks constructed using different eigenmodes. $s  = 7$ for $p = 0.5,$ $s  = 11$ for $p = 0.05,$ and $s  = 12$ for $p = 0.005$ respectively.}
	\label{fig:Det_prob}
		\vspace*{-0.2 cm}
\end{figure}

We also compare our results to data-driven FDI attacks proposed in prior work \cite{Kim2015,YuBlind2015} in Figure~\ref{fig:Compare_prev_work_14bus}. Here in, attack~1 is constructed according to Theorem~\ref{thm:opt_attack}. Attack~2 is constructed using the entire estimated subspace, i.e., $\av_2 = \widehat{\Um}_N \cv_N,$ where the elements of $\cv_N$ are set to $c_i = \sqrt{\frac{\tau}{N \widehat{\omega}_i/\widehat{\mu}_i}}, i = 1,\dots,N.$ Note that $\cv_N$ is adjusted to satisfy $||\Delta \thetav||^2_2  \geq \tau.$
As expected, the detection probability of attack~1 is significantly lower compared to attack~2. 

\begin{figure}[!t]
	\centering
	\includegraphics[width=0.4\textwidth]{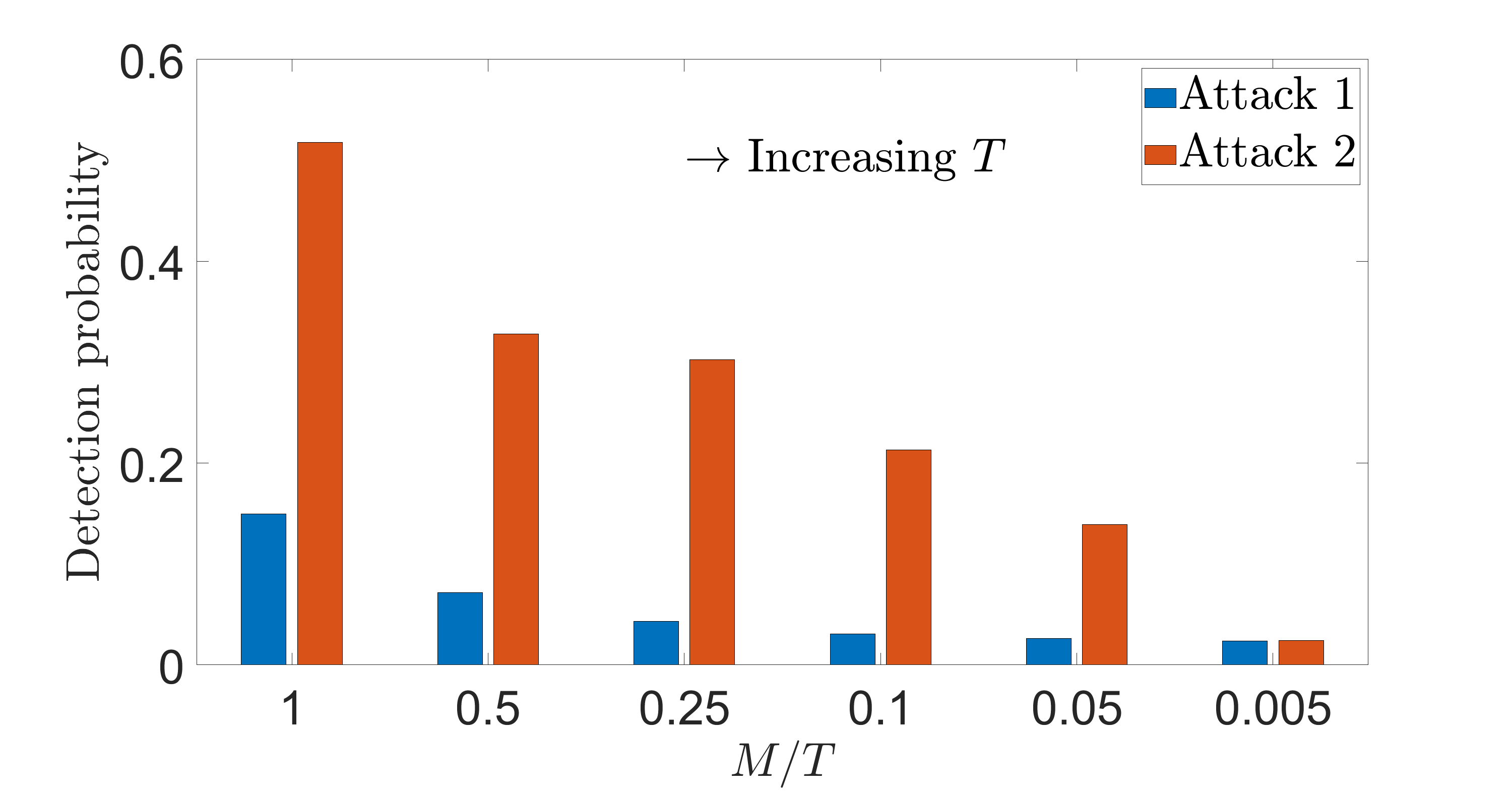}
	\caption{Attack detection probability as a function of the measurement time window. Attack 1 : optimal data-driven FDI attack (Theorem~\ref{thm:opt_attack}), Attack 2: FDI attack constructed using the entire estimated subspace. }
	\label{fig:Compare_prev_work_14bus}
		\vspace*{-0.2 cm}
\end{figure}

\begin{figure}[!t]
	\centering
	\includegraphics[width=0.4\textwidth]{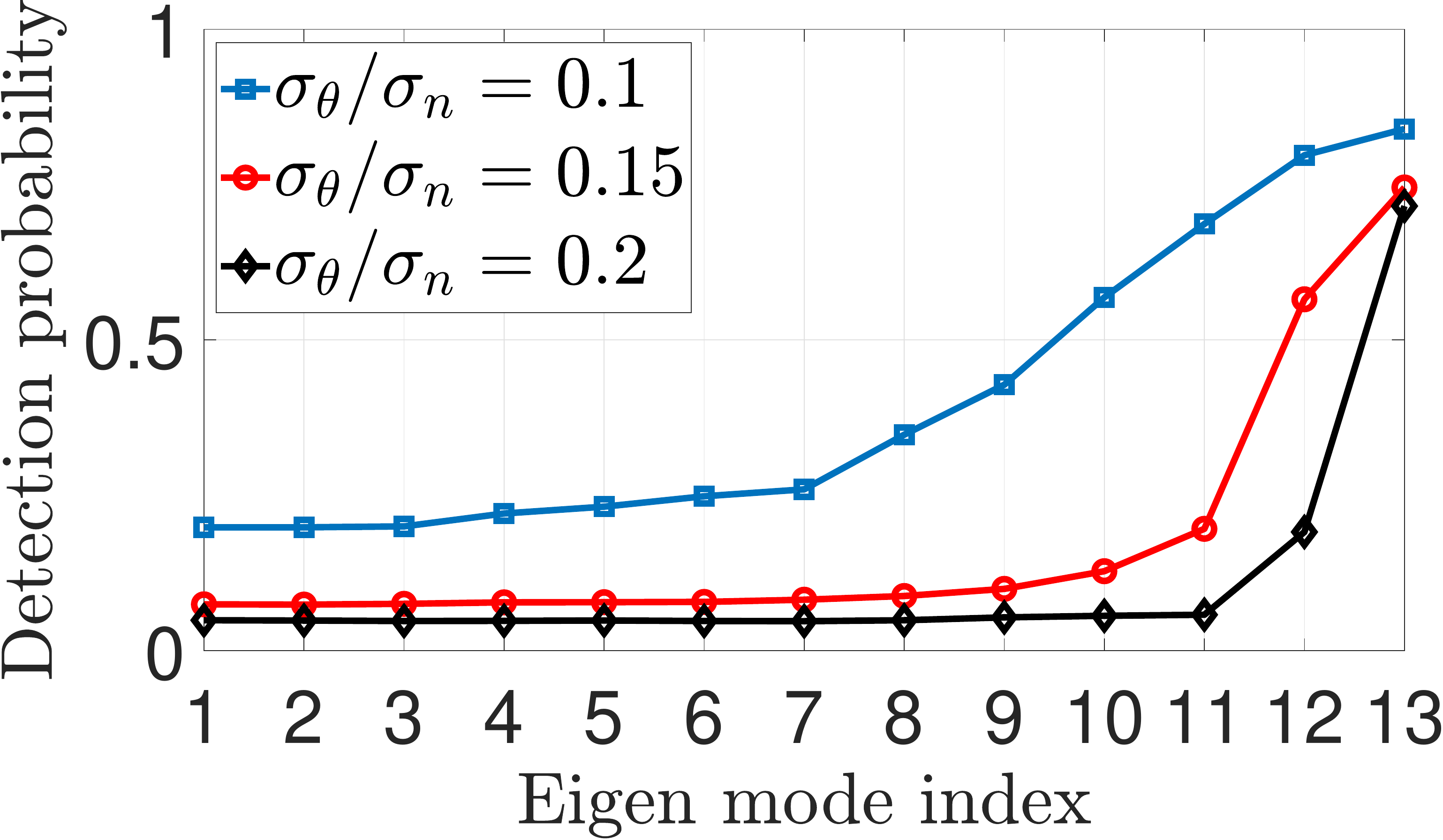}
	\caption{Attack detection probability for $M/T = 0.5$ and different values of $\sigma_{\theta}/\sigma_n.$ }
	\label{fig:Detection_prob_statevar}
\end{figure}

We also examine the algorithm's performance for different values of $\sigma_{\theta}$. We vary $\sigma_{\theta}/\sigma_{n}$ and examine the attack detection probability. In Fig.~\ref{fig:Detection_prob_statevar},  we observe that as $\sigma_{\theta}/\sigma_{n}$ increases, the detection probability decreases. This is because a higher variation in the system state enables the attacker to estimate $Col(\Hm)$ more accurately. This is also confirmed by our theoretical result -- note that $\mu_i$ increases with an increase in $\sigma_{\theta},$ which in turn results in a more accurate estimate of the basis vectors of $Col(\Hm)$  (note from Appendix~C, equation \eqref{eqn:refhere5}, that $\omega_i$ increases with $\mu_i$). Thus, the attacker can bypass the BDD with a higher probability.

\subsection{Attack Detection Probability Under AC State Estimation}
We also test the robustness of the attacks in bypassing the BDD under a non-linear AC power flow model. We inject the attack vector designed according to Theorem~\ref{thm:opt_attack} (i.e., attacks generated based on the linear model) into measurements derived from an AC power flow model and compute the attack detection probability. We adjust the value of $\tau$ to cause different attack impact (measured in terms of $\eta = \frac{|| \widehat{\thetav}_a - \thetav ||_2}{|| \widehat{\thetav} - \thetav||_2}$). For reference, we also plot the detection probability under the state estimation of the DC power flow model with $\eta = 4.$

\begin{figure}[!t]
	\centering
	\includegraphics[width=0.45\textwidth]{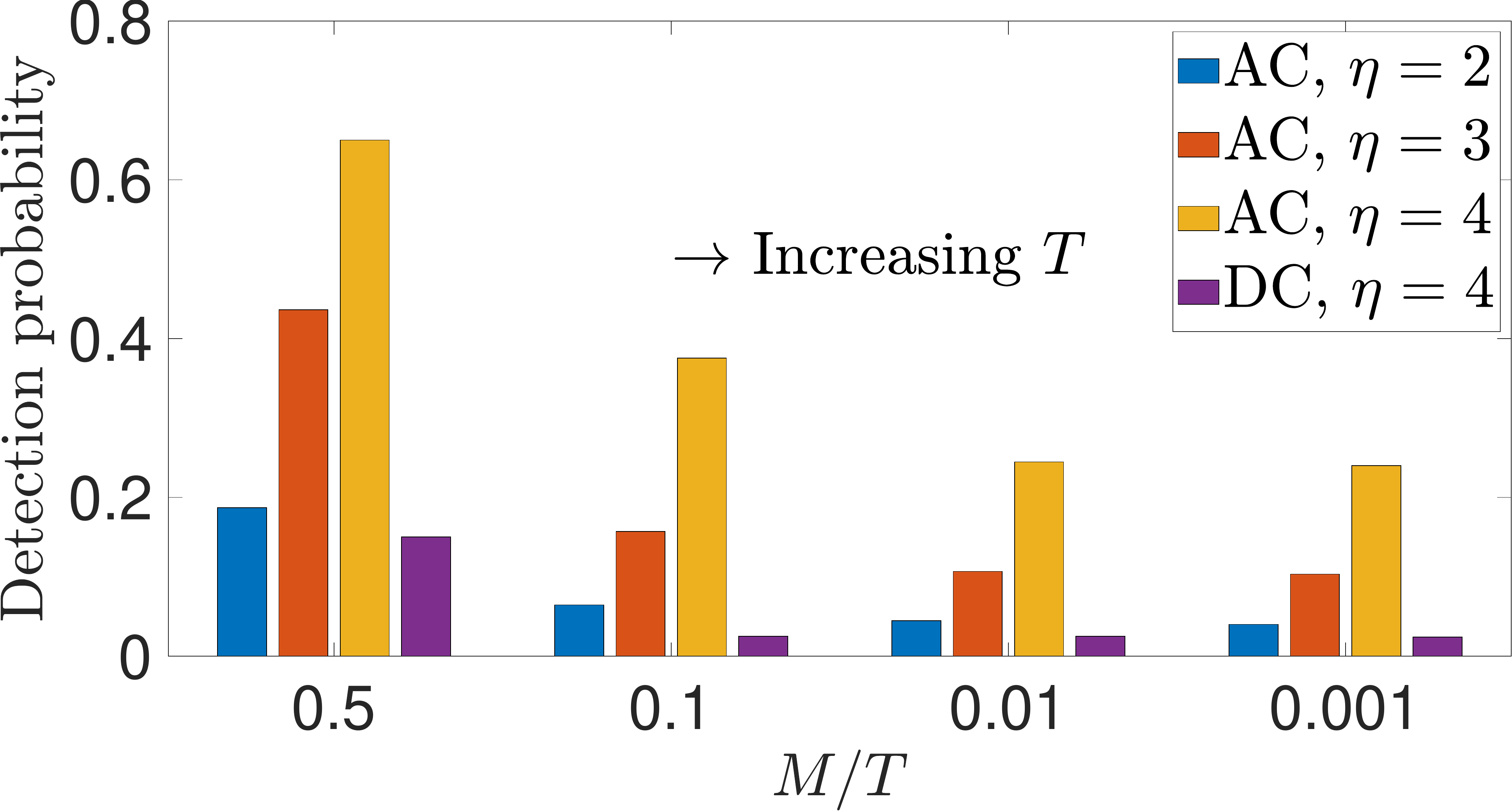}
	\caption{Detection probability of data-driven attack (designed as in Theorem~\ref{thm:opt_attack}) under BDD of AC and DC state estimation for different values of $M/T$ and $\eta$ (attack impact). }
	\label{fig:AC_Detection}
\end{figure}

The results are plotted in Fig.~\ref{fig:AC_Detection}. We observe that attacks designed based on the linear model can bypass the BDD of AC state estimation with a high probability and cause a significant attack impact. For instance, when $M/T = 0.5,$ the attack can achieve $\eta = 2$ while its detection probability remains $0.2.$  We also observe that as the measurement time window $T$ increases, the attack is capable of causing a larger impact while bypassing the BDD with a high probability. This observation is consistent with the findings in \cite{Kim2015,TEIXEIRA2011}, where it was also observed that attacks constructed based on the linear model remain valid under the AC model.

\subsection{Trade-offs in Data-Driven FDI Attacks}
\begin{figure}[!t]
	\centering
	\includegraphics[width=0.45\textwidth]{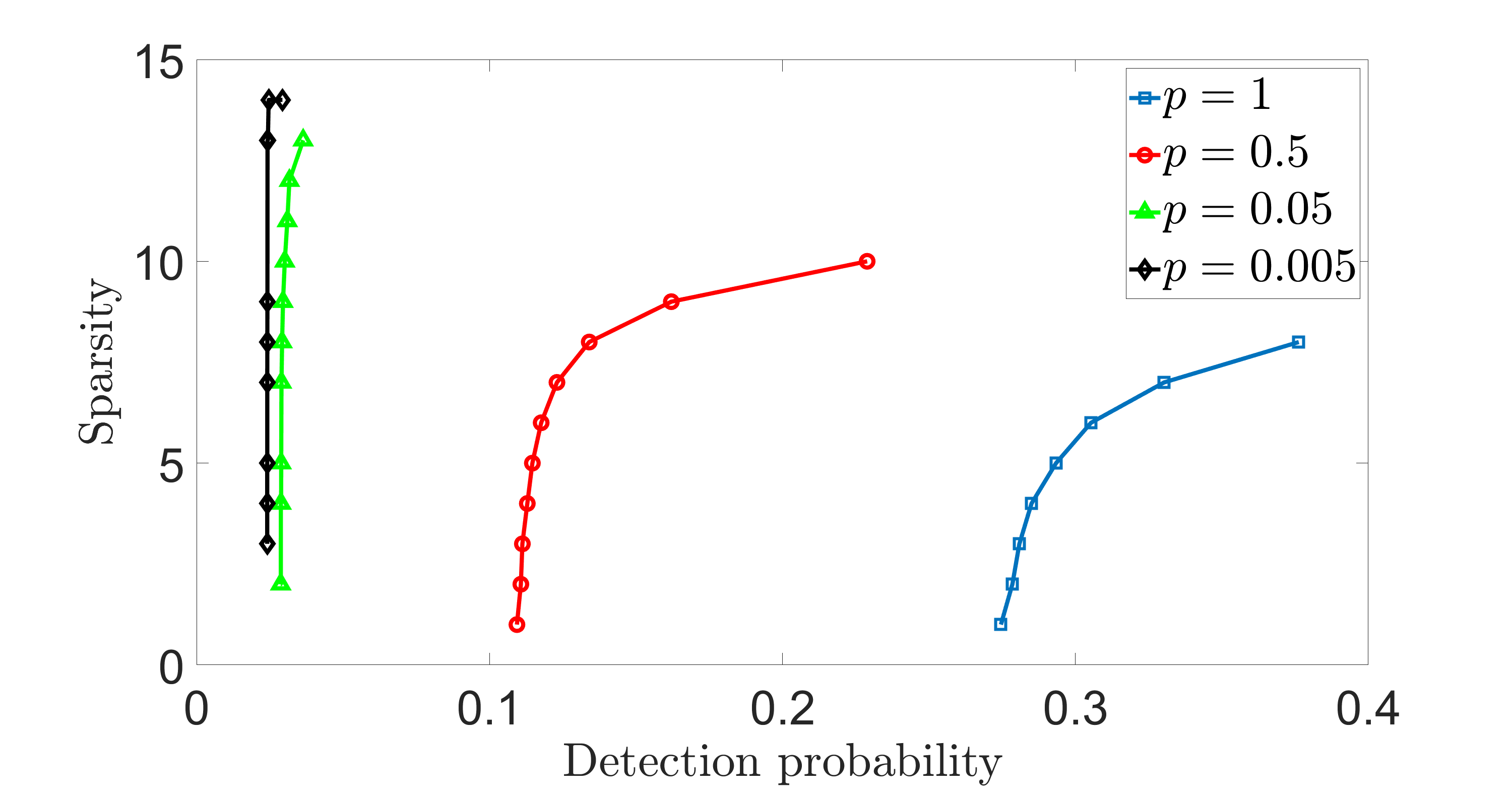}
	\caption{Trade-off between attack sparsity and the attack detection probability.}
	\label{fig:Tradeoff_14bus}
	\vspace*{-0.4 cm}
\end{figure}
Next, we illustrate the trade-off between attack's sparsity and the detection probability in Fig.~\ref{fig:Tradeoff_14bus}. The points on the trade-off curve are obtained by varying $m$ in \eqref{eqn:opt_sparse}. Specifically, we compute an attack vector for each value of $m,$ and then compute the corresponding detection probability and the attack's sparsity. Note that sparsity of the attack vector is equal to $M - K^*_m.$ We repeat the simulations for different training times $T$ (i.e. varying $p$). It can be observed that in the non-asymptotic regime (i.e, small $T$), the attack's sparsity can be enhanced if the attacker can tolerate an increase in the attack detection probability (refer to the red and blue curves).
For large $T$ however, the attacker can simply utilize the entire estimated subspace without having to compromise the attack's detection probability (green and black curves). In practice, the attacker can make use of such trade-off curves to select suitable parameters for the construction of the FDI attack, e.g., based on the available resources.

\subsection{Simulations with Large Bus Systems}
To show our effectiveness of the proposed algorithm in large bus systems, we conduct simulations using IEEE-39 and 118-bus systems. Except for the bus configuration settings, rest of the settings are maintained identical to that of Fig.~\ref{fig:Compare_prev_work_14bus}. We plot the detection probability as a function of the observation time window, and compare it to an approach that uses the entire estimated subspace to construct the FDI attack. In Fig.~\ref{fig:Compare_prev_work_largebus}, it can be observed that under the limited observation time window, the proposed approach significantly reduces the attack detection probability, thus confirming the effectiveness of our approach in these systems as well. 

\begin{figure}[!t]
	\centering
	\includegraphics[width=0.4\textwidth]{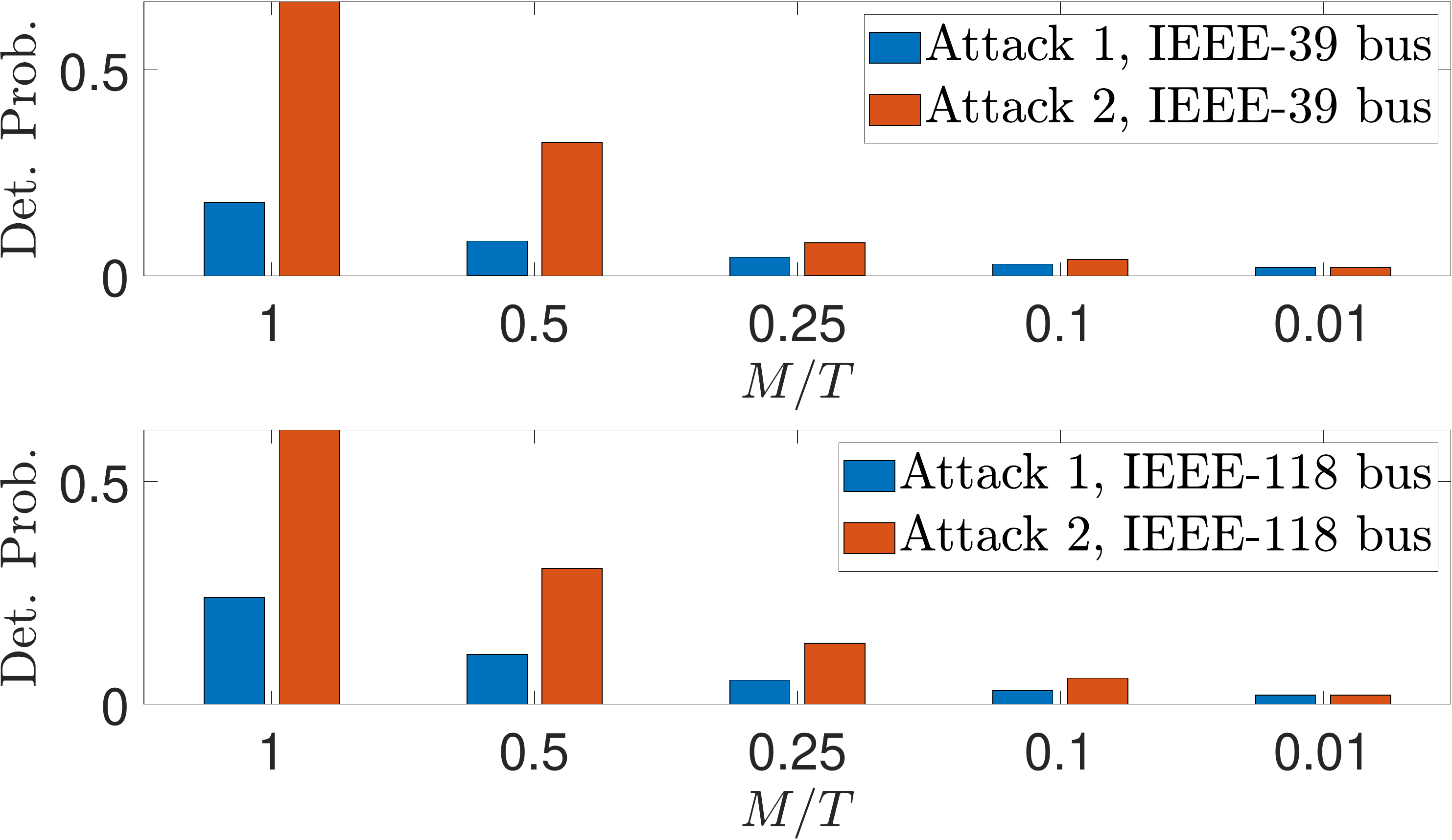}
	\caption{Attack detection probability as a function of the measurement time window for IEEE-39 and 118 bus systems. Attack 1 : optimal data-driven FDI attack (Theorem~\ref{thm:opt_attack}), Attack 2: FDI attack constructed using the entire estimated subspace. }
	\label{fig:Compare_prev_work_largebus}
		\vspace*{-0.2 cm}
\end{figure}

\section{Conclusions and Future Work}
\label{V}
We have studied the construction of data-driven FDI attacks when the attacker has access to measurements from a limited observation time window. We showed that in this regime, the attacker can enhance the BDD-bypass probability by constraining the attack vector to a lower-dimensional subspace spanned by the accurately estimated basis vectors. We used results from RMT spiked model to analyze the algorithm performance. We also characterized an important trade-off between the attacker's ability to bypass the BDD and the sparsity the attack vector. Our framework gives practical guidance to a resource-constrained attacker in designing stealthy FDI attacks. In the future, we will explore how the results of this work can be used to address the defense problem against these attackers (e.g., MTD).

There are several interesting future research directions. First, our framework assumes that the attacker has read access to all the measurements within the considered time window. However, in practice, there may be missing measurements due to communication loss or device malfunctions (see references \cite{Anwar2016}, \cite{TianDataDriven2018}). Studying data-driven FDI attacks with missing measurements under a limited measurement time window setting would require combining RMT results with that of robust PCA techniques, which is an interesting future research direction. Second, the design of data-driven FDI attacks under a limited measurement period setting for a non-linear AC power flow is challenging. While our results evidence that attacks constructed using the linear model can bypass the BDD of AC state estimation, the performance can be further enhanced considering a non-linear model in attack design. Finally, the design of defense strategies against data-driven FDI attacks (such as MTD) based on the understanding of the attacker's capabilities will be a critical problem for power grid operators.

\section*{Appendix A: Undetectable Attacks}
\subsection*{Part I: Proof of Lemma~1}
First recall that the residual vector is given by
\begin{align}
\rv& = \zv- \Hm \widehat{\thetav} = (\Id - \Km) \zv, \label{eqn:refhere1}
\end{align}
where in \eqref{eqn:refhere1}, we have substituted $\widehat{\thetav}  =  (\Hm^T  \Hm)^{-1} \Hm^T  \zv $ (note $\Wm = \Id_M$ ) 
and denoted $\Km = {\bf H}({\bf H}^{T}{\bf H})^{-1}{\bf H}^{T}.$ Further substituting $\zv = \Hm \thetav+\nv,$ we obtain,
\begin{align}
\rv& = (\Id - \Km) (\Hm \thetav+\nv) =  (\Id - \Km) \nv, \label{eqn:refhere2}
\end{align}
where \eqref{eqn:refhere2} follows since $(\Id - \Km) \Hm = {\bf 0}.$

Now consider the residual for measurements with FDI attack $\zv_a = \zv+\av.$
The residual denoted by ${\bf r}_a$ is given by:
\begin{align}
{\bf r}_a &=({\bf I}-{\bf K}){\bf z}_{a} \nonumber \\
& \stackrel{(a)}{=} ({\bf I}-{\bf K}) ({\bf n}+{\bf a}) \nonumber \\
& \stackrel{(b)}{=} ({\bf I}-{\bf U}_N{\bf U}_N^{T})({\bf n}+{\bf a}), \label{eqn:refhere3}
\end{align}
where in $(a),$ we have once again use $(\Id - \Km) \Hm = {\bf 0},$ and in $(b),$ we have used the fact that $\Km$ can be decomposed as ${\bf K}={\bf U}_N {\bf U}_N^{T}.$ Since the noise is Gaussian, $\|{\bf r}_a \|_2^2$ follows a non central chi-square distribution with $M-N$ degrees of freedom and noncentrality parameter $\nu$ given by
\begin{align}
\mathbb{E} \LSB \|{\bf r}_a \|_2^2 \RSB = \nu = {\bf a}^{T} {\bf a} -{\bf a}^{T} {\bf U}_N {\bf U}_N^{T}\widehat{\bf U}_s {\bf a}_s.
\end{align}
In particular, for a data-driven FDI attack of the form ${\bf a}=\widehat{\bf U}_s{\bf c}_s,$ we have 
\begin{align}
\nu = {\bf c}^{T}_s{\bf c}_s-{\bf c}^{T}_s \widehat{\bf U}_s^{T}{\bf U}_N {\bf U}_N^{T}\widehat{\bf U}_s {\bf c}_s.
\end{align}
For $M,T \to \infty, M/T = c,$ using the result of Theorem~\ref{thm:spike_result}, the matrix $\widehat{\bf U}_s^{T}{\bf U}_N {\bf U}_N^{T}\widehat{\bf U}_s$  converges to a diagonal matrix whose diagonal elements are given by 
$\omega_i$ defined in  \eqref{eqn:omega_def}. Since $N$ is assumed to be fixed (and finite) and independent of $M$ and $T$  (see Theorem~\ref{thm:spike_result}), we obtain,
\begin{align*}
	{\bf c}_s^{T} \widehat{\bf U}_s^{T}{\bf U}_N{\bf U}_N^{T}\widehat{\bf U}_s{\bf c}_s  - {\bf c}_s^{T}  \Omegam_s  {\bf c}_s \asto 0.
\end{align*}
Further, using the result $| \mu_i - \widehat{\mu}_i| \asto 0, i = 1,\dots,s$ (Theorem~\ref{thm:spike_result}), from continuous mapping theorem \cite{Bill86}, it follows that 
\begin{align*}
| \omega_i - \widehat{\omega}_i| \asto 0, i = 1,\dots,s.
\end{align*} 
Note that $\omega_i$ is a continuous function $\mu_i$ (see  \eqref{eqn:omega_def}).

\subsection{Part II: Proof of Lemma~3} 
Since $\widehat{\thetav}  =  (\Hm^T  \Hm)^{-1} \Hm^T  \zv ,$ it follows that 
$\Delta \widehat{\thetav}$ is given by
\begin{align}
\Delta \widehat{\thetav} &=  (\Hm^T  \Hm)^{-1} \Hm^T  (\zv^a - \zv) \nonumber \\
\Delta \widehat{\thetav} & =   (\Hm^T  \Hm)^{-1} \Hm^T  \widehat{\Um}_s \cv_s. \label{eqn:refhere4}
\end{align}
From \eqref{eqn:refhere4}, we obtain,
\begin{align*}
|| \Delta \widehat{\thetav} ||_2 & = \cv^T_s   \widehat{\Um}_s^T \Hm  (\Hm^T  \Hm)^{-2} \Hm^T  \widehat{\Um}_s \cv_s \nonumber \\
& \stackrel{(a)}{=} \cv^T_s \widehat{\Um}_s^T \Um_N \Dm^{-1} \Um^T_N \widehat{\Um}_s \cv_s, \\
& \stackrel{(b)}{=} \sigma^2_{\theta} \cv^T_s \widehat{\Um}_s^T \Um_N \Mm^{-1} \Um^T_N \widehat{\Um}_s \cv_s.
\end{align*}
where $(a)$ follows since $\Hm  (\Hm^T  \Hm)^{-2} \Hm^T = \Um^T_N  \Dm^{-1} \Um^T_N.$ Here in, $\Dm = \text{diag} (d_1,\dots,d_N),$ where $\{ d_i \}^N_{i = 1}$ are the first $N$ eigenvalues of $\Hm \Hm^T$ (in decreasing order). In $(b),$ recall that $\Mm = \text{diag} (\mu_1,\dots,\mu_N)$, where $\mu_i$ are the eigenvalues of  $\sigma^2_{\theta} \Hm \Hm^T.$ Similar to the proof of Lemma~1 (Part I), it can be shown that 
\begin{align*}
 \cv^T_s   \widehat{\Um}_s^T \Um_N \Mm^{-1} \Um^T_s \widehat{\Um}_N \cv_s - \widehat{\sigma}^2_{\theta} \cv^T_s \widehat{\Mm}^{-1} \widehat{\Omegam} \cv_s \asto 0.
\end{align*}

\section*{Appendix B: Proof of Lemma~2}
First, we consider the proof of $1 > \omega_1 \geq \omega_2 \geq,\dots,\geq \omega_s > 0.$

We first show that $0 \leq \omega_i \leq 1, \forall i.$ 
By definition $\mu_i > \sqrt{p}$ for $i = 1,\dots,s.$ For $\mu_i > \sqrt{p}$, we have
$1- p/\mu^2_i > 0.$ Thus, $\omega_i > 0, \forall i.$ 

Also, it is straightforward to note that $1- p/\mu^2_i < 1$ and $1+p/\mu_i > 1.$ Thus, $\omega_i = \frac{1- p/\mu^2_i}{1+p/\mu_i} < 1, \forall i.$

Finally, note that the derivative of $\omega_i$ with respect to $\mu_i$ is given by 
\begin{align}
\frac{d \omega_i}{d \mu_i} = \frac{\mu_i^2 p + 2 \mu_i p+ p^2}{\mu_i (\mu_i+p)} > 0. \label{eqn:refhere5}
\end{align}
where in \eqref{eqn:refhere5}, the inequality follows since all terms in the derivative are positive. 
Thus, we conclude that $\omega_i$ is an increasing function of $\mu_i.$ Since by definition, $\mu_1 \geq \mu_2 \geq,\dots,\geq \mu_s,$ it follows that $1 > \omega_1 \geq \omega_2 \geq,\dots,\geq \omega_s.$

Next, we consider the proof of  $1 > \widehat{\omega}_1 \geq \widehat{\omega}_2 \geq,\dots,\geq \widehat{\omega}_s > 0.$ 
Once again, by definition, we have $\widehat{\lambda}_i >  (1+\sqrt{p})^2$ 
for $i = 1,\dots,s.$  It can be verified from Theorem~\ref{thm:spike_result} that for $\widehat{\lambda}_i >  (1+\sqrt{p})^2,$ $\widehat{\mu}_i > \sqrt{p}.$ 
Thus, $1 > \widehat{\omega}_1 \geq \widehat{\omega}_2 \geq,\dots,\geq \widehat{\omega}_s > 0$ can be proved by arguments identical to the previous case (i.e., the proof of $1 > \omega_1 \geq \omega_2 \geq,\dots,\geq \omega_s > 0$).

\section*{Appendix C: Proof of Theorem~1}
 Note that optimization problem \eqref{eqn:opt_final} can be rewritten as  
 \beqa
\dsp \min_{\cv}    \sum^s_{i = 1}  (1-\widehat{\omega}_i) c^2_i, \label{eqn:opt_eqv}  \ 
 \text{s.t.}   \  \sum^s_{i = 1} \widehat{\sigma}^2_{\theta} \LB \frac{\widehat{\omega}_i}{ \widehat{\mu}_i} \RB c^2_i  \geq \tau. 
 \eeqa 
By a simple replacement of the variable $y_i = c^2_i,$ \eqref{eqn:opt_eqv} becomes 
 \beqa
\dsp \min_{\yv}    \sum^s_{i = 1}  (1-\widehat{\omega}_i) y_i, \label{eqn:opt_eqv_LP} \
 \text{s.t.}   \  \sum^s_{i = 1} \widehat{\sigma}^2_{\theta} \LB \frac{\widehat{\omega}_i}{ \widehat{\mu}_i} \RB y_i  \geq \tau. \nonumber 
\eeqa 
Note that \eqref{eqn:opt_eqv_LP} is a linear programming (LP) problem. Since the coefficients of the objective function as well as the constraints are positive (see Appendix~B), the optimal solution of \eqref{eqn:opt_eqv_LP} must satisfy the constraint with equality, i.e., $\sum^s_{i = 1} \LB \frac{\widehat{\omega}_i}{ \widehat{\mu}_i} \RB y_i  = \tau$. Thus, we can replace the inequality constraint of \eqref{eqn:opt_eqv_LP} with equality. We perform one more change of variable as 
$$y_i = \kappa_i \LB \frac{\tau}{ \widehat{\sigma}^2_{\theta} (\widehat{\omega}_i/\widehat{\mu}_i)} \RB, i = 1,\dots,s.$$
The LP \eqref{eqn:opt_eqv_LP} along with replacing the constraint with equality now becomes
 \beqa
\dsp \min_{{ \bf \kappa}}    \sum^s_{i = 1}  \LB \frac{1-\widehat{\omega}_i}{ \widehat{\sigma}^2_{\theta} ( \widehat{\omega}_i/\widehat{\mu}_i)} \RB \kappa_i , \label{eqn:opt_eqv2} \
 \text{s.t.} \   \sum^s_{i = 1}  \kappa_i = 1 \nonumber.
\eeqa 
It can be verified that the coefficients of the objective function $$ \frac{1- \widehat{\omega}_i}{\widehat{\sigma}^2_{\theta} ( \widehat{\omega}_i/\widehat{\mu}_i)}$$ is a decreasing function of $\widehat{\mu}_i.$ (This can be verified by differentiating the coefficient terms with respect to $\mu_i$ and noting that the derivative is negative.) Since $\mu_1 \geq \mu_2 \geq\dots \geq \mu_s$, we have, 
$$ \frac{1- \widehat{\omega}_1}{ \widehat{\sigma}^2_{\theta}( \widehat{\omega}_1/\widehat{\mu}_1)} \geq \frac{1- \widehat{\omega}_2}{ \widehat{\sigma}^2_{\theta} (\widehat{\omega}_2/\widehat{\mu}_2)} \geq \dots \geq \frac{1- \widehat{\omega}_s}{ \widehat{\sigma}^2_{\theta}(\widehat{\omega}_s/\widehat{\mu}_s)}.$$
The solution to \eqref{eqn:opt_eqv2} is thus given by 
$\kappa_1 = 1, \kappa_2 = \kappa_3= \dots = \kappa_s = 0.$ Hence, the solution to \eqref{eqn:opt_eqv} becomes $$c_1 = \sqrt{\frac{\tau}{ \widehat{\sigma}^2_{\theta}(\widehat{\omega}_i/\widehat{\mu}_i)}}$$ and $c_2 = c_3= \dots = c_s = 0.$

\bibliographystyle{IEEEtran}
\bibliography{bibliography}

\end{document}

%% file: macros.tex
\def\LSB{\left[}        
\def\RSB{\right]}       
\def\LB{\left(}         
\def\RB{\right)}        

\newcommand{\asto}{\overset{\rm a.s.}{\longrightarrow}}

\newfont{\bbb}{msbm10 scaled 500}

\newfont{\bb}{msbm10 scaled 1100}

\newcommand{\RR}{\mbox{\bb R}}

\newcommand{\av}{{\bf a}}

\newcommand{\cv}{{\bf c}}

\newcommand{\nv}{{\bf n}}

\newcommand{\rv}{{\bf r}}

\newcommand{\uv}{{\bf u}}

\newcommand{\xv}{{\bf x}}
\newcommand{\yv}{{\bf y}}
\newcommand{\zv}{{\bf z}}

\newcommand{\Am}{{\bf A}}

\newcommand{\Dm}{{\bf D}}

\newcommand{\Hm}{{\bf H}}
\newcommand{\Id}{{\bf I}}

\newcommand{\Km}{{\bf K}}

\newcommand{\Mm}{{\bf M}}

\newcommand{\Um}{{\bf U}}
\newcommand{\Wm}{{\bf W}}

\newcommand{\Xm}{{\bf X}}
\newcommand{\Ym}{{\bf Y}}

\newcommand{\epsilonv}{\hbox{\boldmath$\epsilon$}}

\newcommand{\thetav}{\hbox{\boldmath$\theta$}}

\newcommand{\Lambdam}{\hbox{\boldmath$\Lambda$}}

\newcommand{\Sigmam}{\hbox{\boldmath$\Sigma$}}

\newcommand{\Omegam}{\hbox{\boldmath$\Omega$}}

\newcommand{\diag}{{\hbox{diag}}}

\newcommand{\defines}{{\,\,\stackrel{\scriptscriptstyle \bigtriangleup}{=}\,\,}}